\documentclass[a4paper,UKenglish,cleveref, autoref, thm-restate]{lipics-v2021}

\title{Model-checking parametric lock-sharing systems against regular constraints}

\author{Corto Mascle}{Universit\'e de Bordeaux, France}{}{}{}
\author{Anca Muscholl}{Universit\'e de Bordeaux, France}{}{}{}
\author{Igor Walukiewicz}{Universit\'e de Bordeaux, CNRS, France}{}{}{}
\authorrunning{C. Mascle, A. Muscholl and I. Walukiewicz} 

\Copyright{Corto Mascle, Anca Muscholl and Igor Walukiewicz} 



\nolinenumbers 

\EventEditors{Guillermo A. P\'{e}rez and Jean-Fran\c{c}ois Raskin}
\EventNoEds{2}
\EventLongTitle{34th International Conference on Concurrency Theory (CONCUR 2023)}
\EventShortTitle{CONCUR 2023}
\EventAcronym{CONCUR}
\EventYear{2023}
\EventDate{September 18--23, 2023}
\EventLocation{Antwerp, Belgium}
\EventLogo{}
\SeriesVolume{279}
\ArticleNo{24}

\usepackage[utf8]{inputenc}
\usepackage{xcolor}
\usepackage{hyperref}
\usepackage{todonotes}
\usepackage{extarrows}
\definecolor{bluegray}{RGB}{160,200,200}
\usepackage[notion,quotation]{knowledge}

\usepackage{bookmark}
\definecolor{Blue Sapphire}{HTML}{003050} 
\definecolor{Gamboge}{HTML}{ee9b00}
\definecolor{Ruby Red}{HTML}{9b2226}

\IfKnowledgePaperModeTF{
}{
	\knowledgestyle{intro notion}{color={Ruby Red}, emphasize}
	\knowledgestyle{notion}{color={Blue Sapphire}}
	\hypersetup{
		colorlinks=true,
		breaklinks=true,
		linkcolor={Blue Sapphire}, 
		citecolor={Blue Sapphire}, 
		filecolor={Blue Sapphire}, 
		urlcolor={Blue Sapphire},
	}
}
\IfKnowledgeCompositionModeTF{
	\knowledgeconfigure{anchor point color={Ruby Red}, anchor point shape=corner}
	\knowledgestyle{intro unknown}{color={Gamboge}, emphasize}
	\knowledgestyle{intro unknown cont}{color={Gamboge}, emphasize}
	\knowledgestyle{kl unknown}{color={Gamboge}}
	\knowledgestyle{kl unknown cont}{color={Gamboge}}
}{
}

\knowledge{notion}
 | scope
 | scopes

\knowledge{notion}
| origin

\knowledge{notion}
| ultimately held

\knowledge{notion}
| tree automaton
| tree automata

\knowledge{notion}
| starting node

\knowledge{notion}
| soundness

\knowledge{notion}
| comparable

\knowledge{notion}
| infinite descending chain

\knowledge{notion}
| F1
| F1-5
| F1-3,5

\knowledge{notion}
| F2

\knowledge{notion}
| F3

\knowledge{notion}
| F4

\knowledge{notion}
| F5

\knowledge{notion}
| F4.1

\knowledge{notion}
| F4.2

\knowledge{notion}
| F4.3

\knowledge{notion}
| witnesses
| witnessing
| witnessed

\knowledge{notion}
| Lock-sharing system
| lock-sharing system
| LSS

\knowledge{notion}
| dynamic Lock-sharing system
| dynamic lock-sharing system
| DLSS

\knowledge{notion}
| DLSS verification problem

\knowledge{notion}
| pushdown DLSS

\knowledge{notion}
| pushdown tree automata
| pushdown tree automaton

\knowledge{notion}
| right-resetting

\knowledge{notion}
| syntactic H-labeling
| syntactic $H$-labeling

\knowledge{notion}
| consistent order labeling

\knowledge{notion}
| run
| runs

\knowledge{notion}
| enumeration

\knowledge{notion}
| process-consistent

\knowledge{notion}
| branch-consistent

\knowledge{notion}
| eventuality-consistent

\knowledge{notion}
| semantically correct

\knowledge{notion}
| syntactically correct

\knowledge{notion}
| uncolored

\knowledge{notion}
| ev-trace

\knowledge{notion}
| label-consistent

\knowledge{notion}
| recurrence-consistent

\knowledge{notion}
| order-consistent

\knowledge{notion}
| locally consistent

\knowledge{notion}
| fairness-consistent

\knowledge{notion}
| unmatched get

\knowledge{notion}
| abstraction tree
| abstraction trees

\knowledge{notion}
| locally lower than

\knowledge{notion}
| limit
| limit configuration
| limit configurations

\knowledge{notion}
| local configuration
| local configurations

\knowledge{notion}
| local transition
| local transitions

\knowledge{notion}
| local run

\knowledge{notion}
| available

\knowledge{notion}
| unmatched
| matched

\knowledge{notion}
|controllable
|Controllable

\knowledge{notion}
|uncontrollable
|Uncontrollable

\knowledge{notion}
| free
| taken

\knowledge{notion}
| global configurations
| global configuration
| configurations
| configuration

\knowledge{notion}
| control strategy
| strategy
| strategies
| Strategies
| Strategy

\knowledge{notion}
| local strategy
| local strategies

\knowledge{notion}
| uses two locks
| 2LSS

\knowledge{notion}
| locally live

\knowledge{notion}
| regular objective
| regular objectives
| Regular objectives

\knowledge{notion}
| process-fair
| process-fairness
| fair
| \textbf{process-fair}
| Fairness

\knowledge{notion}
| strongly process-fair

\knowledge{notion}
| yields a $P$-deadlock
| deadlock@P

\knowledge{notion}
| yields a partial deadlock
| partial deadlock
| partial deadlocks

\knowledge{notion}
| circular deadlock
| circular deadlocks

\knowledge{notion}
| pattern@inf
| infinitary@pat
| infinitary patterns
| infinitary pattern
| Infinitary patterns

\knowledge{notion}
| weak pattern@inf
| weak@inf 

\knowledge{notion}
| sound
| Soundness

\knowledge{notion}
| exclusive
| Exclusive

\knowledge{notion}
| neutral

\knowledge{notion}
| stair decomposition

\knowledge{notion}
| stair pattern

\knowledge{notion}
| strong pattern@inf
| strong patterns@inf
| strong@inf

\knowledge{notion}
| switching

\knowledge{notion}
| global deadlock
| global deadlocks

\knowledge{notion}
| G_{\mathit{Inf}}

\knowledge{notion}
| \text{\sc{Blocks}}

\knowledge{notion}
| pattern
| Patterns
| patterns
| finite@pat
| finitary patterns
| Finitary patterns
| finitary@pat
| finitary pattern

\knowledge{notion}
| trace
| Trace

\knowledge{notion}
| strong pattern
| strong patterns
| strong@pat

\knowledge{notion}
| weak pattern
| weak patterns
| weak@pat

\knowledge{notion}
| risky

\knowledge{notion}
| nested

\knowledge{notion}
| admits a pattern

\knowledge{notion}
| behavior@two

\knowledge{notion}
| regular verification problem
| Regular verification problem

\knowledge{notion}
| partial deadlock objective

\knowledge{notion}
| generalised deadlock objective

\knowledge{notion}
| partial deadlock problem

\knowledge{notion}
| process deadlock problem
| Process deadlock problem
| process deadlocks

\knowledge{notion}
| generalised deadlock control problem

\knowledge{notion}
| deterministic Emerson-Lei automaton
| deterministic Emerson-Lei automata
| Emerson-Lei automaton
| ELA
| DELA
\newif\ifappendix
\newrobustcmd\introinrestatable[1]{%
	\ifappendix%
	\kl{#1}%
	\else%
	\intro{#1}%
	\fi%
}

\usepackage[utf8]{inputenc}
\usepackage[T1]{fontenc}
\usepackage{mathtools}
\usepackage{amssymb}
\usepackage{thm-restate}
\usepackage{xspace}
\usepackage{tikz}
\usepackage{array}
\usepackage{changepage}
\usepackage{hyperref}
\usepackage{xcolor}
\usepackage{algorithm}
\usepackage{algpseudocode}

\definecolor{light-gray}{gray}{0.75}
\definecolor{darkgreen}{RGB}{0,200,0}
\definecolor{lightpurple}{RGB}{220,0,220}
\usepackage{algpseudocode}
\usepackage[capitalise]{cleveref}
\usepackage{cite}

\usetikzlibrary{arrows,calc,automata,shapes,positioning}
\tikzset{AUT style/.style={>=angle 60,initial text= ,every edge/.append style={thick},every state/.style={thick,minimum size=15,inner sep=0.5}}}

\makeatletter
\def\namedlabel#1#2{\begingroup
	\def\@currentlabel{#2}%
	\label{#1}\endgroup
}
\makeatother

\renewcommand{\b}{\beta}
\renewcommand{\d}{\delta}
\newcommand{\e}{\varepsilon}

\newcommand{\g}{\gamma}

\newcommand{\n}{\nu}

\newcommand{\s}{\sigma}

\renewcommand{\t}{\tau}
\renewcommand{\th}{\theta} 

\newcommand{\w}{\omega}

\newcommand{\W}{\Omega}
\renewcommand{\S}{\Sigma}
\renewcommand{\epsilon}{\varepsilon}
\renewcommand{\phi}{\varphi}

\newcommand{\nats}{\mathbb{N}}

\newcommand{\set}[1]{\{ #1 \}}

\newcommand{\tree}{\tau}

\newcommand{\subtree}{\theta}

\newcommand{\vallabel}{\mathit{L}}

\newcommand{\Locks}{\mathit{{\cal L}\!ocks}}
\newcommand{\Var}[1]{\mathit{Var(#1)}}

\newcommand{\bef}{\prec_H}
\newcommand{\aft}{\succ_H}

\newcommand{\greenedge}[1]{\xhookrightarrow{#1}}
\newcommand{\rededge}[1]{\xmapsto{#1}}

\newcommand{\PP}{\mathbb{P}}

\knowledgenewrobustcmd{\BTPP}{BT_{\PP}}
\knowledgenewrobustcmd{\strongedge}[1]{\cmdkl{\greenedge{#1}}}
\knowledgenewrobustcmd{\weakedge}[1]{\cmdkl{\rededge{#1}}}

\knowledgenewrobustcmd{\spat}{\cmdkl{\Longrightarrow}}
\knowledgenewrobustcmd{\wpat}{\cmdkl{\dashrightarrow}}
\knowledgenewrobustcmd{\spatinf}{\cmdkl{\Longrightarrow}}
\knowledgenewrobustcmd{\wpatinf}{\cmdkl{\dashrightarrow}}
\knowledgenewrobustcmd{\Gp}{\cmdkl{G_{\PP}}}
\knowledgenewrobustcmd{\Gu}{\cmdkl{G_{u}}}

\knowledgenewrobustcmd{\BTu}{\cmdkl{BT_{u}}}
\knowledgenewrobustcmd{\FTu}{\cmdkl{FT_{u}}}
\knowledgenewrobustcmd{\PZ}{\cmdkl{Proc_Z}}
\knowledgenewrobustcmd{\pad}[1]{#1^{\cmdkl{\dummy}}}
\knowledgenewrobustcmd{\systemchoices}{\cmdkl{\mathbb{SC}}}
\knowledgenewrobustcmd{\Gblockchain}[2]{\cmdkl{\mathcal{G}}_{#1,#2}}

\makeatletter
\newcommand{\dashover}[2][\mathop]{#1{\mathpalette\df@over{{\dashfill}{#2}}}}
\newcommand{\fillover}[2][\mathop]{#1{\mathpalette\df@over{{\solidfill}{#2}}}}
\newcommand{\df@over}[2]{\df@@over#1#2}
\newcommand\df@@over[3]{%
	\vbox{
		\offinterlineskip
		\ialign{##\cr
			#2{#1}\cr
			\noalign{\kern1pt}
			$\m@th#1#3$\cr
		}
	}%
}
\newcommand{\dashfill}[1]{%
	\kern-.5pt
	\xleaders\hbox{\kern.5pt\vrule height.4pt width \dash@width{#1}\kern.5pt}\hfill
	\kern-.5pt
}
\newcommand{\dash@width}[1]{%
	\ifx#1\displaystyle
	2pt
	\else
	\ifx#1\textstyle
	1.5pt
	\else
	\ifx#1\scriptstyle
	1.25pt
	\else
	\ifx#1\scriptscriptstyle
	1pt
	\fi
	\fi
	\fi
	\fi
}
\newcommand{\solidfill}[1]{\leaders\hrule\hfill}
\makeatother

\knowledgenewrobustcmd{\Owns}{\cmdkl{\text{\sc{Owns}}}}

\knowledgenewrobustcmd{\Inf}{\cmdkl{\text{\sc{Inf}}}}

\newcommand{\aut}{\mathcal{A}}

\knowledgenewrobustcmd{\lss}{\emph{"LSS"}}

\knowledgenewrobustcmd{\trace}[1]{\cmdkl{tr}(#1)}

\newcommand{\sat}{\mathtt{Sat}}

\newcommand{\init}{\mathit{init}}

\newcommand\restrict[2]{{
		\left.\kern-\nulldelimiterspace
		#1
		\vphantom{\big|}
		\right|_{#2}
}}

\renewcommand{\sat}{\models}

\newcommand{\Aa}{\mathcal{A}}
\newcommand{\Bb}{\mathcal{B}}
\newcommand{\Cc}{\mathcal{C}}

\newcommand{\Oo}{\mathcal{O}}
\newcommand{\Pp}{\mathcal{P}}

\newcommand{\Ss}{\mathcal{S}}

\newcommand{\es}{\emptyset}
\newcommand{\incl}{\subseteq}
\newcommand{\act}[1]{\xlongrightarrow{#1}}

\newcommand{\wh}[1]{\widehat{#1}}

\newcommand{\Aal}{\Aa^l}

\newcommand{\nop}{\mathtt{nop}}
\newcommand{\dummy}{\square}
\newcommand{\get}[1]{\mathtt{get}_{#1}}
\newcommand{\rel}[1]{\mathtt{rel}_{#1}}
\newcommand{\spawn}[1]{\mathtt{spawn(}#1\mathbf{)}}
\newcommand{\new}{\mathtt{new}}
\newcommand{\ppop}{\mathtt{pop}}
\newcommand{\ppush}{\mathtt{push}}
\newcommand{\reset}{\mathtt{reset}}
\newcommand{\instr}{\mathtt{instr}}
\newcommand{\LFP}{\mathsf{LFP}}
\newcommand{\GFP}{\mathsf{GFP}}
\newcommand{\sig}{\mathit{sig}}
\newcommand{\keeps}{\text{\textit{keeps}}}
\newcommand{\evkeeps}{\text{\textit{ev-keeps}}}
\newcommand{\avoids}{\text{\textit{avoids}}}
\newcommand{\evavoids}{\text{\textit{ev-avoids}}}
\newcommand{\Pre}{\text{\textit{Pre}}}
\newcommand{\Pres}{\Pre^*}

\renewcommand{\d}{\delta}

\newcommand{\first}{\mathit{first}}
\newcommand{\nnext}{\mathit{next}}
\newcommand{\phil}{\mathit{phil}}

\newcommand{\pinit}{p_{\init}}

\newcommand{\Proc}{\mathit{Proc}}
\newcommand{\ch}{\mathit{ch}}
\newcommand{\tk}{\mathit{tk}}
\newcommand{\spw}{\mathit{sp}}

\newcommand{\ar}{\mathit{ar}}
\newcommand{\op}{\mathit{op}}
\newcommand{\Op}{\mathit{Op}}

\newcommand{\PTIME}{\text{\sc Ptime}}

\newcommand{\EXPTIME}{\text{\sc Exptime}}

\begin{document}

\appendixfalse

\maketitle

\begin{abstract}
	In parametric lock-sharing systems processes can spawn new
	processes to run in parallel, and can create new locks. 
	The behavior of every process is given by a pushdown automaton.
	We consider infinite behaviors of such systems under strong process fairness
	condition. 
	A result of a potentially infinite execution of a system is a limit
	configuration, that is a potentially infinite tree. 
	The verification problem is to determine if a given system has a limit
	configuration satisfying a given regular property. 
	This formulation of the problem encompasses verification of reachability as well
	as of many 	liveness properties. 
	We show that this verification problem, while undecidable in general, is 
	decidable for nested lock usage.

	We show \EXPTIME-completeness of the verification problem. 
	The main source of complexity is the number of parameters in the spawn
	operation. 
	If the number of parameters is bounded, our algorithm works in \PTIME\ for
	properties expressed by parity automata with a fixed number of ranks. 

\end{abstract}

\section{Introduction}
Locks are a widely used concurrency primitive.
They appear in classical programming languages
such as Java, as well as in recent ones such as Rust. 
The principle of creating shared objects and protecting them by 
mutexes (like the “synchronized” paradigm in Java) requires
\emph{dynamic lock creation}. 
The challenge is to be able to analyze programs with dynamic creation
of threads \emph{and} locks.

Our system model is based on Dynamic  Pushdown Networks (DPNs) as introduced
in~\cite{BouajjaniMT05}, where processes are pushdown systems that can
spawn new processes.
The DPN model was extended in~\cite{LammichMW09} by adding
synchronization through a fixed number of locks.
Here we take a step further and allow dynamic lock creation:
when spawning a new process, the parent process can pass some of its locks, and
new locks can be created for the new thread.
This way we model recursive programs with creation of threads and locks.
We call such systems \emph{dynamic lock-sharing systems} (DLSS).

The focus in both \cite{BouajjaniMT05} and \cite{LammichMW09} is
computing the $\Pres$ of a regular set of configurations, and they achieve
this by extending suitably the saturation technique from~\cite{BouajjaniEM97}.
Here we consider not only reachability but also infinite behaviors of DLSS under
fairness conditions.
For this we propose a different approach than saturation from
\cite{BouajjaniMT05,LammichMW09} as saturation is not 
suited to cope with liveness properties.

We show that verifying regular properties of DLSS is decidable if every process
follows \emph{nested lock usage}.
This means that locally every process acquires and releases the locks
according to the stack discipline.
Nested locking is assumed in most papers on parametric verification of
systems with synchronization over locks.
It is also considered as good programming practice, sometimes even enforced
syntactically, as in Java through synchronized blocks.

Without any restriction on lock usage we show that our problem is
undecidable, even for finite state processes and reachability
properties that refer to a single process.
Note that our model does not have global variables.
It is well-known that reachability is undecidable already for two pushdown processes with one lock and
one global variable.  
\medskip

\emph{Outline of the paper.}
Our starting point is to use trees to represent configurations of DLSS.
This representation was introduced in~\cite{LammichMW09}.
The advantage  is that it does not require to talk about
interleavings of local runs of processes.
Instead it represents every local run as a left branch in a tree and
the spawn operations as branching to the right.
At each computation step one or two nodes are added below a leaf of the current
configuration. 
Thus, the result of a run of DLSS is an infinite tree that we call a
\emph{limit configuration}.
Our first observation is that it is easy to read out from a limit configuration
of a run if the run is strongly process-fair (Proposition~\ref{prop:fair}). 

An important step is to characterize those trees that are limit
configurations of runs of a given \emph{finite state} DLSS, namely where every
process is a finite state system.
This is done in Lemma~\ref{lem:limit}.
To deal with lock creation this lemma refers to the existence of some global acyclic
relation on locks.
We show that this global relation can be recovered from local orderings in
every node of the configuration tree (Lemma~\ref{lem:limitlocal}).
Finally, we show that there is a nondeterministic  B\"uchi tree
automaton verifying
all the conditions of Lemmas~\ref{lem:limit} and~\ref{lem:limitlocal}.
This is the desired tree automaton recognizing limit configurations of process-fair runs. 
Our verification problem is solved by checking if there is a tree satisfying the
specification and accepted by this automaton.
This way we obtain the upper bound from Theorem~\ref{thm:main-finite}.
Surprisingly the size of the Büchi automaton is
linear in the size of DLSS, and exponential only in the
\emph{arity} of the DLSS, which
is the maximal number of locks a process can access.
For example, in the dining philosophers setting  (cf.~Figure~\ref{fig:philosophers})  the arity is $3$, as every philosopher has
access only to its left and right forks, implemented as locks; and there is one
more fork to close the cycle.

The extension of our construction to pushdown processes requires one
more idea to get an optimal complexity. 
In this case, ensuring that the limit tree represents a computation requires
using pushdown automata. 
Hence, the B\"uchi tree automaton as described in the previous paragraph becomes
a pushdown B\"uchi automaton on trees.
The emptiness of pushdown B\"uchi tree automata is \EXPTIME-complete, which is
an issue as the automaton constructed is already exponential in the size of the 
input. However, we 
observe that the automata we obtain are right-resetting, since new
threads are spawned with empty pushdown.
Intuitively, the pushdown is needed only on left paths of the
configuration tree to check correctness of local runs. 
A right-resetting automaton resets its stack each time it goes to the right child.
We show that the emptiness of right-resetting parity pushdown tree automata can be
checked in \PTIME\ if the biggest rank in the parity condition is fixed (if it
is not fixed then the problem is at least as complex as  solving parity games).
This gives the upper bound from Theorem~\ref{thm:main-pushdown}.

Finally, we obtain the matching lower bound by proving
\EXPTIME-hardness of checking if a process of the DLSS has an infinite
run (Proposition~\ref{prop:EXPTIME-hard}).
This holds even for finite state processes.
We also show that even for finite state processes the DLSS verification problem
is undecidable if we allow arbitrary usage of locks (Theorem~\ref{th:undec}). 
\medskip

\emph{Related work.}
Parametrized verification has remained an active research area for almost three
decades~\cite{germansistla92,BloJacKha15,Abdulla:2018aa}. 
It has brought a steady stream of works on parametric systems with locks.
As already mentioned, the first directly relevant paper is~\cite{BouajjaniMT05} introducing
Dynamic Pushdown Networks (DPNs).
These consist of pushdown processes with spawn but no locks.
The main idea is to represent a configuration as a sequence of process
identifiers, each identifier followed by a stack content. 
Computing $\Pres$ of a regular set of configurations is decidable by extending
the saturation technique from~\cite{BouajjaniEM97}.

An important step is made in~\cite{LammichMW09} where the authors introduce
a tree representation of configurations. 
This is essentially the same representation as we use here.
They extend DPNs by a fixed set of locks, and show how to adapt the saturation
technique to compute $\Pres$ in this case.
Their result is an \EXPTIME\ decision procedure for verifying reachability of a
regular set of configurations.
This work has been extended to incorporate join operations~\cite{GawlitzaLMSW11}, or
priorities on processes~\cite{DiazT17}.
Our work extends~\cite{LammichMW09} in two directions: it adds lock creation, and
considers liveness properties.
It is not clear how one could extend saturation methods to deal
with liveness properties.
 
The saturation method has been adapted to DPNs with lock creation in the recent
thesis~\cite{Kenter22}. 
The approach relies on hyperedge replacement grammars, and gives decidability without 
complexity bounds. 
Our liveness conditions can express this kind of reachability conditions.


Actually, the first related paper to deal with lock creation is
probably~\cite{YasukataT016}.
The authors consider a model of higher-order programs with
spawn, joins, and lock creation.
Apart from nested locking, a new restriction of scope safety is imposed. 
Under these conditions, reachability of pairs of states is shown to be decidable.

The works above have been followed by implementations~\cite{Lammich11,YasukataT016,DiazT17}.
In particular \cite{DiazT17} reports on verification of several substantial
size programs and detecting an error in xvisor~\cite{xvisor}.

In all the works above nested locking is
assumed. In~\cite{KahIvaGup05} the interest of nested locking is
underlined by showing that reachability with two pushdown processes
using locks is undecidable in general, but it is decidable for nested locking.
There are only few related works without this assumption.
The work~\cite{Kahlon09} generalizes nested locking to bounded lock-chain condition, and
shows decidability of reachability for two pushdown processes. 
In~\cite{LammichMSW13} the authors consider contextual locking where arbitrary locking
may occur as long as it does not cross procedure boundaries. 
This condition is incomparable with nested locking.

Finally, we comment on shared state and global variables.
These are not present in the above models because reachability for two pushdown
processes with one lock and one global variable is already undecidable. 
There is an active line of study of multi-pushdown systems where shared state is
modeled as global control. 
In this model decidability is recovered by imposing restrictions on stack usage
such as bounded context switching and variations thereof~\cite{QadeerR05,TorreMP07,TorreNP20,AkshayGKR20}.
Observe that these are restrictions on global runs, and not on local runs of
processes, as we consider here. 
Another approach to recover decidability is to have shared state but no
locks~\cite{Hague11,EsparzaGM16,FortinMW17,MuschollSW17}.
Finally, there is a very interesting model of threaded
pools~\cite{BaumannMTZ20, BaumannMTZ22}, without locks, where verification
is decidable once again assuming bounded context switching.
But the complexity of this model is as high as Petri net
coverability~\cite{BaumannMTZ22}. 


\medskip

\emph{Structure of the paper.}
The next section presents the main definitions and results. 
The main proof for finite state processes is outlined is Sections~\ref{sec:limit} and~\ref{sec:automaton}.
Section~\ref{sec:pushdown} describes the extension to pushdown processes.
All missing proofs are included in the appendix.



\section{Definitions and results}\label{sec:defs}

A dynamic lock-sharing system is a set of processes, each
process has access to a set of locks and can spawn other processes. 
A spawned process can inherit some locks of the spawning process and can also
create new locks. 
All processes run in parallel.
A run of the system must be fair, meaning that if a
process can move infinitely many times then it eventually does. 

More formally, we start with a finite set of process identifiers $\Proc$.
Each process identifier $p\in \Proc$ has an arity $\ar(p)\in\nats$ telling how many
locks the process uses.
The process can refer to these locks through the variables $\Var{p}=\set{x_1^p, \ldots,
x_{\ar(p)}^p}$. 
At each step a process can do one of the following operations:
\begin{align*}
        \Op(p) =& \set{\nop}\cup\set{\get{x}, \rel{x} \mid x \in \Var{p}} \\
        & \cup \set{\spawn{q, \sigma} \mid q \in \Proc, \sigma : \Var{q} \to (\Var{p}
        \cup \set{\new})}
\end{align*}
Operation $\nop$ does nothing. 
Operation $\get{x}$ acquires the lock designated by $x$, while $\rel{x}$
releases it. 
Operation $\spawn{q,\sigma}$ spawns an instance of process $q$ where every variable
of $q$ designates a lock determined by the substitution $\s$; this can be a lock
of the spawning process or a new lock, if $\s(x^q)=\new$.
We require that the mapping $\s$ is \emph{injective} on $\Var{p}$.
This is important for the definition of nested stack usage. 

\AP A ""dynamic lock-sharing system"" ("DLSS" for short) is a tuple 
\begin{equation*}
        \Ss=(\Proc,\ar,(\aut_p)_{p\in\Proc}, \pinit,\Locks)
\end{equation*}
where $\Proc$, and $\ar$ are as described above.
For every process $p$,  $\aut_p$ is a transition
system describing the behavior of $p$.
Process $\pinit\in \Proc$ is the initial process.
Finally, $\Locks$ is an infinite pool of locks.

Each transition system $\aut_p$ is a tuple $(S_p, \S_p, \delta_p,\op_p,
\init_p)$ with $S_p$ a finite set of states, $\init_p$ the initial state, $\S_p$
a finite alphabet, $\delta_p : S_p \times \S_p \to S_p$ a \emph{partial}
transition function, and $\op_p:\S_p\to\Op(p)$ an assignment of an operation
to each action. 
We require that the $\S_p$ are pairwise disjoint, and define $\S = \bigcup_{p
\in \Proc} \S_p$. 
We write $\op(b)$ instead of $\op_p(b)$ for $b\in \S_p$, as $b$
determines the process $p$.

For simplicity, we require that $\pinit$ is of arity $0$. In
particular, process $\pinit$ has no $\get{}$ or $\rel{}$ operations. 

An example in Figure~\ref{fig:philosophers} presents a "DLSS" modeling an
arbitrary number of dining philosophers. 
The system can generate a ring of arbitrarily many philosophers, but can also
generate infinitely many philosophers without ever closing the ring.
\begin{figure}
   \begin{align*}
    \pinit:&\quad \spawn{\first,\new,\new}\\
    \first(x_l,x_r):&\quad \spawn{\phil,x_l,x_r};\spawn{\nnext,x_r,\new,x_l}\\
    \nnext(x_l,x_r,x_{\text{lfirst}}):&\quad \text{or}
      \begin{cases}
        \spawn{\phil,x_l,x_{\text{lfirst}}}&\\
        \spawn{\phil,x_l,x_r};\spawn{\nnext,x_r,\new,x_{\text{lfirst}}}&
    \end{cases}\\
    \phil(x_l,x_r):&\quad \text{repeat-forever or}
      \begin{cases}
        \get{x_l};\get{x_r};\text{eat};\rel{x_r};\rel{x_l}&\\
        \get{x_r};\get{x_l};\text{eat};\rel{x_l};\rel{x_r}&
      \end{cases}
  \end{align*}
   \caption{Dining philosophers: process $\first$ starts the first
   philosopher and an iterator process $\nnext$ starts successive philosophers. The
   forks, modeled as locks, are passed via variables $x_l$ and $x_r$. The third variable
   $x_{\text{lfirst}}$ of $\nnext$ is the left fork of the first philosopher used
   also by the last philosopher. The system is nested as $\phil$ takes and releases
   forks in the stack order. The arity of the system is $3$.}
   \label{fig:philosophers}
\end{figure}

A configuration of $\Ss$ is a tree representing the runs of all active
processes.
The leftmost branch represents the run of the initial process $\pinit$, the left
branches of nodes to the right of the leftmost branch represent runs of processes
spawned by $\pinit$ etc.
So a leaf of a configuration represents the current situation of a process that is
started at the first ancestor above the leaf that is a right child.
A node of a configuration is associated with a process, and tells in
what state the process is, which locks are available to it, and which of
them it holds.

\AP More formally, a ""configuration"" is a, possibly
infinite,
tree $\tree\incl\set{0,1}^*$, with each  node $\n$ labeled by a tuple
$(p,s,a, L,H)$
where $p\in\Proc$ is the process executing in $\n$, $s\in \S_p$ the state
of $p$, $a \in \S_p$ the action $p$ executed at $\n$, or $\bot\not\in\S$ if
$\n$ is a root,
$L:\Var{p}\to\Locks$ an assignment of locks to variables of $p$, and
$H\incl L(\Var{p})$ the set of locks $p$ holds before executing $a$.
We use  $p(\n)$, $s(\n)$, $a(\n)$, $L(\n)$ and $H(\n)$ to address the components of the
label of $\n$.
For ease of notation we will write $\Var{\n}$ instead of
$\Var{p(\n)}$. 

\AP We write $H(\tree)$ for the set of locks ""ultimately held"" by some process in
$\tree$, that is,
$H(\tree)=\set{\ell : \text{for some $\nu$, $\ell \in H(\nu')$ for all
    $\nu'$ on the leftmost path from $\nu$}}$. 
If $\t$ is finite this is the same as to say that $H(\tree)$ is the union of $H(\nu)$ over
all leaves $\n$ of $\t$.



The initial "configuration" is the tree $\t_\init$ consisting only of the root
$\e$ labeled by $(\pinit,\init_p,\bot,\es,\es)$.
Suppose that $\n$ is a leaf of $\tree$ labeled by $(p,s,b,L,H)$,  and there is a transition $s\act{a} s'$
for some $s'$ in $\aut_p$.
A transition between two "configurations" $\tree\act{\n,a}\tree'$ is defined
by the following rules.
\begin{itemize}
  \item If $\op(a)=\spawn{q,\s}$ then $\tree'$ is obtained from $\tree$ by adding two
  children $\n0,\n1$ of $\n$.
  The label of the left child $\n0$ is $(p,s',a,L,H)$. 
  The label of the right child $\n1$ is $(q,\init_q,\bot,L',\es)$ where
  $L'(x^q)=L(\s(x^q))$ if $\s(x^q)\not=\new$ and $L'(x^q)=\ell_{\n,x^q}$ is a fresh lock, otherwise.
  \item Otherwise, $\tree'$ is obtained from $\tree$ by adding a left child $\n0$ to $\n$.
  The label of $\n0$ must be of the form $(p,s',a,L,H')$ subject to the following
  constraints:
  \begin{itemize}
    \item If $\op(a)=\nop$ then $H'=H$,
    \item If $\op(a)=\get{x}$ and $L(x)\not\in H(\tree)$ then $H'=H\cup\set{L(x)}$,
    \item If $\op(a)=\rel{x}$ and $L(x)\in H$ then $H'=H \setminus  \set{L(x)}$.
  \end{itemize}
\AP Note that we do not allow a process to acquire a lock it already holds, or
release a lock it does not have. We call this property
""soundness"".
\end{itemize}

\AP A ""run"" is a (finite or infinite) sequence of "configurations" 
$\tree_0\xrightarrow{\nu_1, a_1} \tree_1 \xrightarrow{\nu_2, a_2}
\cdots$.
As the trees in a run are growing we can define the ""limit configuration"" of that run as
its last configuration if it is finite, and as the limit of its configurations if it
is infinite. 

\begin{remark}
  Note that in a "run", at every moment  distinct variables of a process are
  associated with distinct locks: $L(\n_i)(x) \not= L(\n_i)(y)$ for
  all $x,y \in \Var{\nu}$ with $x \not= y$.
\end{remark}

\begin{remark}\label{rem:labels-finite}
  The labels $L$ and $H$ can be computed out of the other three labels in the tree
  just following the transition rules. 
  We could have defined "configurations" as trees with only three labels
  $(p,s,a)$, but we preferred to include $L$ and $H$ for readability.
  Yet, later we will work with tree automata recognizing configurations
  and there it will be important that the labels come from a finite set. 
\end{remark}

\AP A configuration $\t$ is ""fair"" if for no leaf $\nu$ 
there is a transition $\t\act{\n,a}\t'$ for some $a$ and $\t'$. 
We show that this compact definition of fairness captures strong process
fairness of runs.
Recall that a run is ""strongly process-fair"" if whenever from some position in the run  a
process is enabled infinitely often then it moves after this position. 

\begin{restatable}{proposition}{PropFair}
	\label{prop:fair}
	Consider a "run" $\tree_0\xrightarrow{\nu_1, a_1} \tree_1 \xrightarrow{\nu_2, a_2} \cdots$
	and its "limit configuration" $\t$.
	The run is "strongly process-fair" if and only if $\t$ is "fair".
\end{restatable}

\noindent
\textbf{Objectives.} Instead of using some specific temporal logic we
stick to a most general specification formalism and use regular tree
properties for specifications.
A ""regular objective"" is given by a nondeterministic tree automaton
$\Bb$ over
$\S\cup\set{\bot}$, which defines a language of accepted "limit configurations".
The trees we work with can have nodes of rank $0$, $1$, or $2$.
So we suppose that the alphabet is partitioned into $\S_0$, $\S_1$ and
$\S_2$.
The nondeterministic transition function reflects this with
$\d(q,a)\incl\set{\top}$ if $a\in\S^0$, $\d(q,a)\incl Q$ if $a\in\S_1$, and 
$\d(q,a)\incl  Q\times Q$ if $a\in\S_2$.
A run of the automaton on a tree $t$ is a labeling of $t$ with states
respecting $\d$. 
In particular if $\n$ is a leaf of $t$ then $\top\in\d(q,a)$, where $q$ is the
state and $a$ is the letter in $\n$.
A run is accepting if for every infinite path the sequence of states on this
path is in the accepting set of the automaton.
We will work with accepting sets given by parity conditions. 
We say that a "configuration" $\t$ satisfies $\Bb$ when $\Bb$ accepts the tree
obtained from $\t$ by restricting only to action labels. 

Regular objectives can express many interesting properties.\label{page:properties}
For example, ``for every instance of process $p$ its run is in a regular language
$\Cc$''.
Or more complicated ``there is an instance of $p$ with a run in a regular
language $\Cc_1$ and all the instances of $p$ have runs in the language
$\Cc_2$''.
Of course, it is also possible to talk about boolean combinations of such
properties for different processes.
Observe that the resulting automaton $\Bb$ for these kinds of properties can be
a parity automaton with ranks $1,2,3$ (properties of sequences can be expressed
by B\"uchi automata, and rank $3$ is used to implement existential
quantification on process instances).

Regular objectives can express deadlock properties.
Since we only consider "process-fair" "runs", a finite branch
in a "limit configuration" indicates that a process is blocked forever after some
point. 
Hence, we can express properties such as ``there is an instance of $p$ that is blocked forever
after a finite run in a regular language $\Cc$''. We can also express that all
branches are finite, which is equivalent to a global deadlock since we are considering only "process-fair" "runs".

Reachability properties are also expressible with regular objectives.
We can check simultaneous reachability of several states in different branches,
for instance ``there is a reachable configuration in which some process $p$
reaches $s$ while some process $p'$ reaches $s'$''. 
There are ways to do it directly, but the shortest argument is through a small
modification of the "DLSS". 
We can simply add transitions to stop processes non-deterministically in desired
states: adding new $\nop$ transitions from $s$ and $s'$ to new deadlock states.
Using ideas from~\cite{LammichMSW13} we can also check reachability of a regular
set of configurations.

Going back to our dining philosophers example from Figure~\ref{fig:philosophers},
we can see also other types of properties we would like to express. 
For example, we would like to say that there are finitely many philosophers in
the system. 
This can be done simply by saying that there are not infinitely many spawns in
the limit configuration.
(In this example it is equivalent to saying that there is no branch turning
infinitely often to the right.)
Then we can verify a property like ``if there are finitely many processes in the
system and some philosopher eats infinitely often then all philosophers eat
infinitely often''.
This property holds under process-fairness, as philosophers release both their forks after eating.


\begin{definition}[""DLSS verification problem""]
  Given a "DLSS" $\Ss$ and a "regular objective" $\Bb$ decide if there is a "process-fair" "run"
  of $\Ss$ whose "limit configuration" $\t$ satisfies $\Bb$.
\end{definition}

Without any further restrictions we show that our problem is
undecidable: 

\begin{restatable}{theorem}{ThmUndec}
	\label{th:undec}
  The  "DLSS verification problem" is undecidable.
  The result holds even if the "DLSS" is finite-state and  every
  process uses at most 4 locks. 
\end{restatable}

This result is obtained by creating an unbounded chain of processes simulating a Turing machine. Each process memorizes the content of a position on the tape, and communicates with its neighbours by interleaving lock acquisitions. The trick for processes to exchange information by interleaving lock acquisitions was already used in~\cite{KahIvaGup05}, and requires a non-nested usage of locks.

The situation improves significantly if we assume nested usage of locks.

\begin{definition}
  A process $\aut_p$ is  ""nested"" if it takes and releases locks according to a stack
  discipline, i.e., for all $x,y \in \Var{p}$, for
  all paths $s_0 \xrightarrow{a_1} \cdots
  \xrightarrow{a_n} s_n$ in $\aut_p$, with $op(a_1)=\get{x}$,
  $op(a_n)=\rel{x}$, $op(a_m) \not=\rel{x}$ for all $m<n$:
  if $op(a_i) = \get{y}$ for some $i<n$ then there exists $i<k<n$ such that $op(a_k) =
  \rel{y}$.  
  A "DLSS" is nested if all its processes are nested.
\end{definition}

We can state the first main result of the paper. 
Its proof is outlined in the next two sections. 
\begin{restatable}{theorem}{ThmMainFinite}
	\label{thm:main-finite}
  The "DLSS verification problem" for "nested" "DLSS"  is   \EXPTIME-complete.
  It is in \PTIME\ when the number of priorities in the specification
  automaton, and the maximal arity of processes are fixed. 
\end{restatable}

We can extend this result to "DLSS" where transition systems
of each process 
are given by a pushdown automaton (see definitions in
Section~\ref{sec:pushdown}).
The complexity remains the same as for finite state processes.

\begin{restatable}{theorem}{ThmMainPush}
	\label{thm:main-pushdown}
  The "DLSS verification problem" for "nested" "pushdown DLSS" is
  \EXPTIME-complete. It is in 
  \PTIME\ when the number of priorities in the specification
  automaton, and the maximal arity of processes is fixed. 
\end{restatable}


\section{Characterizing limit configurations}
\label{sec:limit}

A "configuration" is a labeled tree. 
We give a characterization of such trees that are "limit configurations"
of a "process-fair" 
"run" of a given "DLSS".
In the following section we will show that the set of "limit configurations" of
a given "DLSS" is a regular tree language, which will imply the decidability of our
verification problem. 

\begin{definition}
  Given a "configuration" $\tree$ with nodes $\n,\n'$ and variables
  $x \in \Var{\n}$, $x' \in \Var{\n'}$, we
  write $x \sim x'$ if $\vallabel(\n)(x)=\vallabel(\n')(x')$, so if $x$
  and $x'$ are mapped to the same lock.
  The ""scope"" of a lock $\ell$ is the set
  $\set{\n :\ell \in \vallabel(\n)(\Var{\n})}$.
\end{definition}

\begin{remark}\label{rem:scope}
  It is easy to see that in any configuration, the "scope" of a lock
  is a subtree. 
\end{remark}

We say that a node $\nu$ is labeled by an ""unmatched"" $\get{}$ if
it is labeled by some $\get{x}$ and there is no $\rel{x}$ operation in the leftmost path starting from
$\nu$.
Recall that $H(\tree)$ is the set of locks $\ell$ for which there is
some node $\nu$ with an "unmatched" $\get{x}$  and
$\vallabel(\nu)(x)=\ell$.

We define a relation $\bef$ on $H(\tree)$ by letting $\ell \bef \ell'$
if there exist two nodes $\n,\n'$ such that $\n$ is an ancestor of $\n'$,
$\nu$ is labeled with an "unmatched" $\get{}$ of $\ell$,
and $\n'$ is labeled with a $\get{}$ of $\ell'$.  

After these preparations we can state a central lemma giving a structural
characterization of "limit configurations" of "process-fair" "runs".

\begin{restatable}{lemma}{LemLimit}
	\label{lem:limit}
	A tree $\tree$ is the "limit configuration" of a "process-fair"
  "run" of a "nested" "DLSS" $\Ss$ if and only if
  \begin{description}
    
  \item[""F1""]  The node labels in $\tree$  match the local transitions of
  $\Ss$.
    
  \item[""F2""] For every leaf $\nu$ every possible transition from
      $s(\nu)$ has operation $\get{x}$ for some $x$  with $\vallabel(\nu)(x) \in
      H(\tree)$.
  \item[""F3""] For every lock  $\ell \in H(\tree)$ there are
      finitely many nodes with operations on $\ell$, and there
      is a unique node  labeled with an "unmatched" 
      $\get{}$ of $\ell$.
      
    \item[""F4""] The relation $\bef$ is acyclic.
    
    \item[""F5""] The relation $\bef$ has no infinite descending
        chain. 
  \end{description}
\end{restatable}

Before presenting the proof of the previous lemma note that the main
difficulty is the fact that some locks can be taken and never released.
If $H(\t)=\es$ then from $\t$ we can
easily construct a run with "limit configuration" $\t$ by exploiting
the "nested" lock usage.
This is because any local run can be executed from a configuration
where all locks are available. 

\begin{proof}
	We start with the left-to-right
        implication. Suppose that we have a process-fair run $\tree_0 \xrightarrow{\nu_1,
          a_1} \tree_1 \xrightarrow{\nu_2, a_2} \cdots$ with "limit
        configuration" $\tree$.
	
	With every lock $\ell \in H(\tree)$ we 
        associate the maximal position $m=m_{\ell}$ such that
        $op(a_m)=\get{x}$ and $\vallabel(\nu_m)(x)=\ell$, so the position $m_{\ell}$
        where $\ell$ is acquired for the last time (and never released
        after). 

        It remains to check the  conditions of the lemma. The first one holds by definition of a "run".
        The second condition is due to process fairness and "soundness", since a
        process can always execute  transitions other than acquiring a
        lock, and locks not in $H(\tree)$ are free infinitely often.
 All actions involving $\ell \in H(\tree)$ must happen before
        position $m_\ell$, hence there are finitely many of them.
        Moreover, a lock cannot be acquired and never released more
        than once.
        This shows condition "F3".
        Conditions "F4" and "F5" are both satisfied because
        if $\ell \bef \ell'$ then $m_\ell < m_{\ell'}$.
        Thus $\bef$ is acyclic and it cannot have infinite descending
        chains.

  For the right-to-left implication, let  $\t$ 
  satisfy all conditions of the lemma.
  In order to construct a run from $\tau$ we first build a
  total order $<$ on $H(\t)$ that extends $\bef$ and has no infinite
  descending chain.
  Let $\ell'_0, \ell'_1, \ldots$ be some arbitrary enumeration of
  $H(\t)$ (which exists as $\tau$ is countable, thus so is $H(\t)$). 
  For all $i$ let $\downarrow\ell'_i = \set{\ell' \in H(\t) \mid \ell'
    \bef^+ \ell'_i}$. 
  As $\t$ satisfies condition "F3", the set of nodes that are
  ancestors of a node with an operation on $\ell'_i$ is finite.
  Since additionally by condition "F5" there are no infinite descending chains for
  $\bef$, the set $\downarrow\ell'_i$ is finite as well (by König's lemma).
  As $\bef$ is acyclic by condition "F4", we can chose some strict
  total order $<_i$ on $\downarrow\ell'_i$ that extends $\bef$.
  We define for all $\ell \in H(\t)$ the index $m_\ell = \min\set{i
    \in \nats \mid \ell \in \downarrow\ell'_i}$. 
  Finally, we set $\ell < \ell'$ if either $m_\ell < m_{\ell'}$ or if
  $m_{\ell} = m_{\ell'}$ and $\ell <_{m_\ell} \ell'$.
 By definition  $<$ is a strict total order on
 $H(\t)$ with no infinite descending chains.
 Moreover it is easy to see that if $\ell \bef \ell'$ then $\ell <
 \ell'$.
 This is the case because $\ell \bef \ell'$ and $\ell' \bef^+ \ell_i$
 implies $\ell \bef^+ \ell_i$, so $m_\ell \le m_{\ell'}$. 
  
Using the order $<$ on $H(\t)$ we construct  a process-fair run $\tree_0 \act{+} \tree_1
        \act{+} \cdots$  with $\tree$ as "limit configuration".
	During the construction we maintain the following invariant
  for every $i$:
  \begin{quote}
    There exists $k_i \in \nats$ such that all operations on locks $\ell_j$ with $j<k_i$ 
    are already executed in $\tree_i$ (there is no operation on
    these locks in $\tree \setminus \tree_i$).
    Moreover, all other locks are
    free after executing $\tree_i$: $H_i:=H(\t_i)=\set{\ell_0, \ldots, \ell_{k_i-1}}$.
  \end{quote}

	For $i=0$ the invariant is clearly satisfied as all locks are free ($k_0=0$). 

  For $i>0$ we assume that there is a run $\tree_0 \act{+} \tree_i$ and
  $\tree_i$ satisfies the invariant.
  Thus, all locks $\ell_j$ with $j<k_i$ are "ultimately held" and all other locks are free in $\t_i$.

  We say that a leaf $\n$ of $\tree_i$ is \emph{available} if one of the
  following holds:
  \begin{enumerate}
  \item either there is a descendant $\n' \not=\n$ on the leftmost path
    from $\n$ in $\tree$ with $H(\n')=H(\n)$ in $\t$,
    \item or  the left child $\n'$ of $\n$ in $\t$ is 
    labeled with an "unmatched" $\get{}$ of $\ell_{k_i}$, and
    there is no further operation on $\ell_{k_i}$ in $\tree \setminus
    \tree_i$. 
  \end{enumerate}

	In particular, leaves of $\t$ cannot be available.
  The strategy is to find the smallest available node $\n$  in BFS order, and
  execute the actions on the left path from $\n$ to $\n'$.
  The execution is possible as on this path there are no actions using locks from $H_i$ and all other
  locks are free. 
  Let $\tree_{i+1}$ denote the configuration thus obtained from
  $\tree_i$. 
  The invariant is satisfied after this execution, with
  $H_{i+1}=H_i$ in the first case above, resp.~$H_{i+1}=H_i \cup
  \set{\ell_{k_i}}$ in the second case.

 It remains to show that if a node is a leaf in $\t_i$ for all $i$
 after some point, then it is a leaf in $\t$.
 This shows, in particular, that there always exists some available
 node. 
 
 	Suppose that $\nu$ and $i_0$ are such that $\nu$ is a leaf of $\t_i$ for all $i\geq i_0$. If $\nu$ becomes available at some point then it stays available in all future configurations, and there are finitely many nodes before $\nu$ in the BFS order. 
 	Thus $\nu$ cannot be available in some $\t_i$, as otherwise it
        would eventually be taken.  	
  Note that by the invariant (and "soundness"), no leaf of $\tree_i$ has the
  left child labeled by some $\rel{}$ operation.
Moreover, every leaf $\n$ of $\tree_i$ with left child $\n'$ in
$\tree$ labeled by $\nop$, $\spawn{}$, or by some matched $\get{}$, is available (the latter because we consider "nested" "DLSS").
	Hence, the left child of $\nu$ must be labeled with an "unmatched" $\get{}$ of some $\ell \in H(\t)$. Thus there is some "unmatched" $\get{}$ on a lock of $H(\t)$ that is never executed.
	
	Let $m$ be the minimal index in the enumeration of $H(\t)$
        such that an "unmatched" $\get{}$ of $\ell_m$ in $\t$ is
        never executed. By minimality of $m$, there exists $i_1$ such
        that $m=k_i$ for all $i \geq i_1$. After $i_1$, all operations
        on locks $\ell < \ell_m$ have been executed. Thus, as $<$
        extends $\bef$, all "unmatched" $\get{}$ operations that have
        some descendant in $\t$ with operation on $\ell_m$, have been executed.
 	By the previous argument, the nodes with left child not
        labeled with an "unmatched" $\get{}$ cannot stay leaves
        forever. Hence, all nodes whose left child has some operation on $\ell_m$ eventually become leaves. The ones with matched $\get{}$ or other operations are then available and eventually executed.
 	
 	Hence, after some point the only remaining operations on $\ell_m$ are "unmatched" $\get{}$. Furthermore by the condition "F3" of the lemma there is exactly one. As a result, when it is reached and all other operations on $\ell_m$ have been executed, it becomes available, and is thus eventually executed, contradicting the definition of $m$.
 	
 	This proves that the limit of the run we have constructed is $\t$.	
  Observe finally that the run is
  process-fair because of condition "F2" of the lemma.
\end{proof}

The next lemma is an important step in the proof as it simplifies condition "F4"
of Lemma~\ref{lem:limit}.
This condition talks about the existence of a global order on some locks. 
The next lemma replaces this order with local orders in each of the nodes. 
These orders can be guessed by a finite automaton.

\begin{restatable}{lemma}{LemLimitLocal}\label{lem:limitlocal}
	Suppose that $\tree$ satisfies the first three conditions of
	Lemma~\ref{lem:limit}.
		The relation $\prec_H$ is acyclic
	if and only if
	there is a family of strict total orders 
	$<_\n$ over a subset of variables from $\Var{\n}$
        such that:
	\begin{description}
		\item[""F4.1""] $x$ is ordered by $<_\n$ if and only if $L(\n)(x)\in H(\t)$.
		\item[""F4.2""] if $x <_\nu x'$, $\nu'$ is a child of
		$\nu$, and $y, y' \in \Var{\nu'}$ are such that $x \sim y$ and $x' \sim
		y'$ then $y <_{\nu'} y'$. 
		\item[""F4.3""]  if $x,x'\in  \Var{\nu}$ and $\vallabel(\n)(x) \prec_H \vallabel(\n)(x')$ then $x <_\n x'$.
	\end{description}
\end{restatable}


\section{Recognizing limit configurations}\label{sec:automaton}

Recall that a "configuration" is a possibly infinite tree with five labels
$p,s,a,L,H$.
As we have mentioned in Remark~\ref{rem:labels-finite},  configurations
need actually only three labels $p,s,a$.
The other two can be calculated from the tree.
Hence, "configurations" are labeled trees with node labels  coming
from a finite alphabet. 
Our goal in this section is to define a  tree automaton recognizing
"limit configurations" of "process-fair" "runs" of a given DLSS.

Our plan is to check the conditions ("F1-5") of Lemma~\ref{lem:limit}.
Actually we will check ("F1-3,5") and the conditions of Lemma~\ref{lem:limitlocal} that are
equivalent to "F4" of Lemma~\ref{lem:limit}.

\AP We first observe that since our processes are finite state it is immediate to
construct a nondeterministic tree automaton $\Bb_1$ verifying condition "F1".
This automaton just verifies local constraints between the labeling of a node and the labelings
of its children.
The constraints talk only about the labels $p,s,a$.
The automaton does not need any acceptance condition, every run is
accepting. 
We will say $\t$ is ""process-consistent"" if it is accepted by
$\Bb_1$.

Checking condition "F2" is more complicated because it refers to a set $H(\t)$ of
locks that are "ultimately held" by some process.
Our approach will be to define four types of predicates and color the
nodes of $\t$ with these predicates. 
From a correct coloring of $\t$ it will be easy to read out $H(\t)$.
Then we will show that the correct coloring can be characterized by conditions
verifiable by finite tree automata. 
The coloring will be also instrumental in checking the remaining
conditions "F3", "F4", "F5".


For a "configuration" $\t$, a node $\n$ and a variable $x\in\Var{\n}$ we define
four predicates.

\begin{itemize}
	\item $\n\sat\keeps(x)$ if at $\n$ process $p(\n)$ holds the lock
	$\ell=L(\n)(x)$ and never releases it: $\ell\in H(\n')$ for every left
	descendant  $\n'$ of $\n$.
	\item $\n\sat\evkeeps(x)$ if  $\n\not \sat\keeps(x)$ and there is a descendant
	$\n'$ of $\n$ and a variable $x' \in \Var{\n'}$ with 	$x \sim x'$ and $\n'\sat\keeps(x')$.
	\item $\n\sat\avoids(x)$ if neither $p(\n)$ nor any descendant
          takes
	$\ell=L(\n)(x)$, namely $\ell\not\in H(\n')$ for every descendant $\n'$ of $\n$
	(including $\n$).
	\item $\n\sat\evavoids(x)$ if $\n \not\sat\avoids(x)$ and on every path from $\n$
          there is $\n'$ such that $\n' \sat \avoids(x)$.
\end{itemize}
Observe a different quantification used in $\evkeeps$ and $\evavoids$. 
In the first case we require one $\n'$ to exist, in the second we want that such
a $\n'$ exists on every path. 

The next lemma shows how we can use the coloring to determine $H(\t)$.

\begin{restatable}{lemma}{LemHColoring}
	\label{lemma:H-from-coloring}
	Let $\t$ be a "process-consistent" "configuration". 
	A lock $\ell\in H(\t)$ if and only if there is a node $\n$ of $\t$ and a variable
	$x\in\Var{\n}$  such that $\n\sat\keeps(x)$ and $L(\n)(x)=\ell$.
\end{restatable}

\begin{proof}
	Follows from the definitions, since $\n\sat\keeps(x)$ if and only if $\ell \in H(\n')$ for
	every left descendant $\n'$ of $\n$.
\end{proof}

\AP The above conditions define a ""semantically correct"" coloring
of nodes of a "configuration" $\t$
by sets of predicates
\begin{equation*}
	\Cc(\n) = \set{P(x): x\in\Var{\n}, \n\sat P(x)}
\end{equation*}
where $P(x)$ is one of $\keeps(x),\evkeeps(x),\avoids(x),\evavoids(x)$. 
Observe that the four predicates are mutually exclusive, but it may be also the
case that none of them holds.
We say that a variable $x \in \Var{\n}$ is ""uncolored"" in $\n$ if $\Cc(\n)$
contains no predicate on $x$.

We now describe consistency conditions on a coloring of "configurations" guaranteeing
that a coloring is  "semantically correct".

Before moving forward we introduce one piece of notation. 
A node that is a right child, namely a node of a form $\n1$ is due to
$\spawn{q,\s}$ operation. 
More precisely $\op(\n0)=\spawn{q,\s}$. 
We refer to this $\s$ as $\s(\n1)$.

\AP A coloring of a "configuration" $\t$ is ""branch-consistent"" if 
for every node $\n$ of $\t$ and every variable $x\in\Var{\n}$ the following
conditions are satisfied.
\begin{itemize}
	\item If $\n$ has one successor $\n0$ then $\n0$ inherits the colors from $\n$ except for two
	cases depending on $\op(\n0)$, i.e, the operation used to obtain $\n0$:
	\begin{itemize}
		\item If $\evkeeps(x)$ is in  $\Cc(\n)$ and the operation is $\get{x}$ then $\Cc(\n0)$ must have
		 either $\evkeeps(x)$ or $\keeps(x)$.
		\item If $\evavoids(x)$ is in $\Cc(\n)$ and the operation is $\rel{x}$ then $\Cc(\n0)$ must have
		either $\evavoids(x)$ or $\avoids(x)$.
	\end{itemize}
	\item If $\n$ has two successors $\n0$, $\n1$, and there is no $y$ with
	$\s(\n1)(y)=x$ then $\n0$ inherits $x$ color from $\n$ and there is no
	constraint due to $x$ on colors in $\n1$.
	\item If $\n$ has two successors and $x=\s(\n_1)(y)$ for some $y\in\Var{\n1}$ then 
	\begin{itemize}
		\item If $\keeps(x)$ in $\Cc(\n)$ then $\keeps(x)$ in $\Cc(\n0)$ and $\avoids(y)$
		in $\Cc(\n1)$.
		\item If $\avoids(x)$ in $\Cc(\n)$ then $\avoids(x)$ in $\Cc(\n0)$ and $\avoids(y)$ in $\Cc(\n1)$.
		\item If $\evkeeps(x)$ in $\Cc(\n)$ then either 
			\begin{itemize}
				\item $\evkeeps(x)$ in $\Cc(\n0)$ and either $\avoids(y)$ or
				$\evavoids(y)$ in $\n1$, or
				\item $\evkeeps(y)$ in $\Cc(\n1)$ and either $\avoids(x)$ or
				$\evavoids(x)$ in $\n0$.
			\end{itemize} 
		\item If $\evavoids(x)$ is $\n$ then $\evavoids(x)$ in $\Cc(\n0)$ and  $\evavoids(y)$ in $\Cc(\n1)$.
	\end{itemize}
\end{itemize}

	
	
	


\AP Next we describe when a coloring is ""eventuality-consistent"".
An ""ev-trace"" is a sequence of pairs $(\n_1,x_1),(\n_2,x_2),\dots$ where :
\begin{itemize}
	\item $\n_1,\n_2,\dots$ is a path in $\t$,
	\item $x_i\in\Var{\n_i}$; moreover $x_{i+1}=x_i$ if $\n_{i+1}$ is the left successor
	of $\n_i$, and $\s(\n_{i+1})(x_{i+1})=x_i$ if $\n_{i+1}$ is the right
	successor of $\n_i$.
	\item $\evkeeps(x_i)$ or $\evavoids(x_i)$ is in $\Cc(\n_i)$.
\end{itemize}
Observe that it follows that it cannot be the case that we have $\evkeeps(x_i)$
and $\evavoids(x_{i+1})$ or vice versa. 
A coloring is "eventuality-consistent" if every "ev-trace" in the
coloring of a "configuration" is finite. 

\AP Finally, a coloring is ""recurrence-consistent"" if for every $\n$
and "uncolored" $x\in\Var{\n}$ the lock $\ell=L(\n)(x)$ is taken and released
infinitely often below $\n$.

\AP A coloring is ""syntactically correct"" if it is 
"branch-consistent", "eventuality-consistent", and "recurrence-consistent".
We show that "syntactically correct" colorings characterize "semantically
correct" colorings.
The two implications are stated separately as the statements are slightly different. 

\begin{restatable}{lemma}{LemSemSynt}
	\label{lem:sem-coloring-is-synt}
	If $\t$ is a "limit configuration" and
        $\Cc$ is a "semantically correct" coloring of $\tree$ then
        $\Cc$ is "syntactically correct". 
\end{restatable}

For the other direction, we prove a more general statement without assuming that
$\t$ is a "limit configuration". 
This is important as ultimately we will use the consistency properties to test
if $\t$ is a "limit configuration".

\begin{restatable}{lemma}{LemSyntSem}
	\label{lem:syntactically-correct}
	If $\t$ is a "configuration" and $\Cc$ a "syntactically correct"
        coloring of $\t$,
	then $\Cc$ is "semantically correct".
      \end{restatable}

Having a correct coloring will help us to verify all conditions of
Lemma~\ref{lem:limit}.
Condition "F2" refers to $L(\n)(x)\in H(\t)$. 
We need another labeling to be able to express this.

\AP A ""syntactic H-labeling"" of $\t$ assigns to every node $\n$ a subset
$H^s(\n)\incl\Var{\n}$.
We require the following properties:
\begin{itemize}
	\item For the root $\e$ we have $H^s(\e)=\es$.
	\item If $\n 0$ exists: $x\in H^s(\n0)$ if and only if $x\in
          H^s(\n)$. 
	\item If $\n 1$ exists: $y\in H^s(\n1)$ if and only if either $\s(\n1)(y)=\new$ and $\n1\sat\evkeeps(y)$, or
	$\s(\n1)(y)=x$ and $\n\sat\evkeeps(x)$.
\end{itemize}
It is clear that every configuration tree has a unique $H^s$ labelling.

\begin{restatable}{lemma}{LemHS}
	\label{lem:Hs}
	Let $\t$ be a "process-consistent" "configuration" with 
        "syntactically correct" coloring.
	For every node $\n$ and variable $x\in\Var{\n}$ we have:
        $L(\n)(x)\in H(\t)$ if and only if $x\in H^s(\n)$. 
\end{restatable}

Thanks to Lemma~\ref{lem:Hs} we obtain
\begin{restatable}{lemma}{LemCheckFtwo}
	\label{lem:chekcing-F2}
	Let $\t$ be a  "process-consistent" "configuration" with a "syntactically
	correct" coloring.  
	Condition "F2" of Lemma~\ref{lem:limit} holds for $\t$ if and only if 
	for every leaf $\n$ of $\t$, every possible transition from
        $s(\n)$ has some
	$\get{x}$ operation with $x\in H^s(\n)$.
\end{restatable}

\begin{restatable}{lemma}{LemCheckFthree}
	\label{lem:chekcing-F3}
	Let $\t$ be a  "process-consistent" "configuration" with a "syntactically
	correct" coloring.  
	Then condition "F3" of Lemma~\ref{lem:limit} holds for $\t$.
\end{restatable}


It remains to deal with conditions "F4" and "F5" of Lemma~\ref{lem:limit}.
Condition "F4" is more difficult to check as it requires to find an acyclic
relation with some properties.
Fortunately Lemma~\ref{lem:limitlocal} gives an equivalent condition talking
about a family of local orders $<_\n$ for every node $\n$ of a
"configuration". 
An automaton can easily guess such a family of orders. 
We show that it can also check the required properties. 

\AP A ""consistent order labeling"" assigns to every node $\n$ of $\t$
a total order $<_\n$ on some subset of $\Var{\n}$.
The assignment must satisfy the following conditions for every node $\n$:
\begin{enumerate}
	\item $x$ is ordered by $<_\n$ if and only if $x\in H^s(\n)$,
	\item if $x<_\n x'$ and $x,x'\in\Var{\n 0}$ then $x <_{\n0} x'$,
	\item if $x<_\n x'$, $\n1$ exists, and $\s(\n1)(y)=x$, $\s(\n1)(y')=x'$ 
          then $y<_{\n1} y'$,
	\item if $\n \sat\keeps(x)$ and  $y<_\n x$ then $\n \sat\keeps(y)$ or $\n \sat\avoids(y)$.
\end{enumerate}

\begin{restatable}{lemma}{LemSyntOrd}	
\label{lem:synt-orders}
  Let $\t$ be a  "process-consistent" "configuration" with a "syntactically
	correct" coloring.  
	A family of local orders $<_\n$ is a "consistent order
        labeling" of $\t$ if and only if it 
	satisfies the conditions of Lemma~\ref{lem:limitlocal}.
\end{restatable}

We consider now condition "F5".
\AP We say that a consistent order labeling of $\t$ admits an \intro{infinite descending chain} if there exist
a sequence of nodes $\n_1,\n_2,\dots$ and variables $(x_i)_i, (y_i)_i$
such that for every $i>0$: (i) $\n_i$ is an ancestor of $\n_{i+1}$, (ii) $y_i \sim x_{i+1}$, and (iii) $y_{i}<_{\n_i} x_i$.

\begin{restatable}{lemma}{WellFounded}
	\label{lem:well-founded}
	Let $\t$ be a  "process-consistent" "configuration" with a "syntactically
	correct" coloring.  
	If $\prec_H$ has no infinite descending chain then there is a "consistent order labeling" of $\t$ with no "infinite descending chain". 
	If $\prec_H$ has an infinite descending chain then every "consistent order
	labeling" of $\t$ admits an "infinite descending chain".
\end{restatable}

The next proposition summarizes the development of this section stating that all the
relevant properties can be checked by a B\"uchi tree automaton. 

\begin{restatable}{proposition}{PropAutLim}
	\label{prop:automaton-limit}
	For a given "DLSS", there is a non-deterministic B\"uchi tree automaton $\wh\Bb$ accepting exactly the
	"limit configurations" of "process-fair" "runs" of DLSS. 
	The size of $\wh\Bb$ is linear in the size of the DLSS and exponential in the maximal
	arity of the DLSS.
\end{restatable}

We will show that the previous proposition yields an \EXPTIME\ algorithm. We match it with an \EXPTIME\ lower bound to obtain completeness.

\begin{restatable}{proposition}{PropExpHard}
	\label{prop:EXPTIME-hard}
	The "DLSS verification problem" for "nested" "DLSS" and
	Büchi objective is \EXPTIME-hard.
	The result holds even if the Büchi objective refers to a
	single process.
\end{restatable}

The hardness proof involves a reduction from the problem of
determining whether the intersection of the languages of $k$
deterministic tree automata over binary trees is empty. To achieve
this, we create a "DLSS" that simulates all the tree automata
concurrently.
Each node of the tree in the intersection is simulated by a process,
which encodes a state for each automaton through the locks it holds.
So each process creates two children with whom it shares locks. The
children are able to access the states of the parent by the following technique: Suppose processes $p$ and $q$ share locks $0$ and $1$, and $p$ acquires one lock and retains it indefinitely. In this scenario, $q$ can guess the lock chosen by $p$ and try to acquire the other lock. If $q$ guesses incorrectly, the system deadlocks. However, if the guess is correct, the execution continues, and $q$ knows about the lock held by $p$.

Now we have all ingredients for the proof of Theorem~\ref{thm:main-finite}: 

\begin{proof}[Proof of Theorem~\ref{thm:main-finite}]
  The lower bound follows from Proposition~\ref{prop:EXPTIME-hard}.

  For the upper bound we use the Büchi tree automaton  $\hat{\Bb}$
  recognizing "limit configurations" of the "DLSS"
  (Proposition~\ref{prop:automaton-limit}).

  We build the
  product of $\hat{\Bb}$ with the "regular objective" automaton $\Aa$, which
  is a parity tree automaton.
  From $\hat{\Bb} \times \Aa$ we can obtain with a bit more work an
  equivalent parity tree automaton $\Cc$ with the same number of
  priorities, plus one.
  For this we modify the rank function in order to only store in the
  state the maximal priority seen  between  
  two consecutive occurrences of Büchi accepting states, and make the
  maximal priority visible at the next Büchi state. 
  When the state of the $\hat{\Bb}$ component is not a Büchi state, the priority is odd and lower than all the ones of $\Aa$.

  By Proposition~\ref{prop:automaton-limit}, $\Cc$ is non-empty if and
 only if there exists a "limit 
  configuration" of the system that satisfies the "regular objective"
  $\Aa$.
Moreover, we know that $\hat{\Bb}$ has
  size linear in the size of the "DLSS" and exponential only in the
  maximal arity of processes.
  So $\Cc$ has size that is exponential w.r.t.~the "DLSS" and the
  objective, and polynomial size if the maximal arity is fixed.

  Finally, non-emptiness of $\Cc$ amounts to solve a parity game of
  the same size as $\Cc$: player Automaton chooses transitions of
  $\Cc$, and player Pathfinder chooses the direction (left/right
  child).
  To sum up, we obtain a parity game of exponential size, so solving
  the game takes exponential time since the number of priorities
  is polynomial.
  If both the number of priorities and the maximal arity are fixed,
  the game can be solved in polynomial time. 
\end{proof}


\section{Pushdown systems with locks}\label{sec:pushdown}

Till now every process has been a finite state system. 
Here we consider the case when processes can be pushdown automata.
The definition of a ""pushdown DLSS"" is the same as before but now
each automaton $\Aa_p$ is a deterministic pushdown automaton.

We will reduce our verification problem to the emptiness test of a
nondeterministic pushdown automata on infinite trees. 
These automata will have parity acceptance conditions.
While in general testing emptiness of such automata is \EXPTIME-complete, we
will notice that the automata we construct have a special form allowing to
test
emptiness in \PTIME\ for a fixed number of ranks in the parity condition.

\AP We start by defining ""pushdown tree automata"".
We work with a ranked alphabet $\S=\S_0\cup\S_1\cup\S_2$, so a letter determines
whether a node has zero, one or two children.
Our automaton will be quite standard but for an additional stack instruction.
Apart standard $\ppop$ and $\ppush(a)$, we have a $\reset$
instruction that empties the stack. 
A pushdown tree automaton is a tuple $(Q,\S,\Gamma,q^0,\bot,\d,\W)$, where $Q$ is a finite
set of states, $\S$ an input alphabet, $\Gamma$ a stack alphabet, $q^0\in Q$ an
initial state, $\bot\in\Gamma$ a bottom stack symbol, and $\W:Q\to\set{1,\dots,d}$
a parity condition. Finally, $\d$ is a partial transition function taking as
the arguments the current state $q$, the current input letter $a$, and the current stack
symbol $\g$. The form of transitions in $\d$ depends on the rank of the letter $a$:
\begin{itemize}
	\item For $a\in \S_0$, we have $\d(q,a,\g)=\top$ for a special symbol $\top$. 
	This means that the automaton accepts in a leaf of the tree if $\d$
	is defined.
	\item For $a\in \S_1$, we have $\d(q,a,\g)=(q',\instr)$ where $\instr$ is one of
	the stack instructions.
	\item For $a\in \S_2$, we have $\d(q,a,\g)=((q_l,\instr_l),(q_r,\instr_r))$, so
	now we have two states, going to the left and right, respectively,
	and two separate stack instructions. 
\end{itemize}

A run of such an automaton on a $\S$-labeled tree is an assignment of
configurations to nodes of the tree; each configuration has the form $(q,w)$
where $q\in Q$ is a state and $w\in\Gamma^+$ is a 
sequence of stack symbols representing the stack (top symbol being the leftmost).
The root is labeled with $(q^0,\bot)$. 
The labelling of children must depend on the labeling of the parent according to
the transition function $\d$. 
In particular, if a leaf of the tree is labeled $a$ and has assigned a
configuration $(q,w)$ then $\d(q,a,\g)$ must be defined, where $\g$ is the
leftmost symbol of $w$.
A run is accepting if for every infinite path the sequence of assigned states
satisfies the max parity condition given by $\W$: the maximum of ranks of states
seen on the path must be even.

\AP We say that a "pushdown tree automaton" is ""right-resetting"" if for every
transition $\d(q,a,\g)=((q_l,\instr_l),(q_r,\instr_r))$ 
we have that $\instr_r$ is $\reset$. 

\begin{restatable}{proposition}{PropRightResetting}\label{prop:emptiness-right-reset}
	For a fixed $d$, the emptiness problem for "right-resetting" pushdown tree
	automata with a parity condition over ranks $\set{1,\dots,d}$ can be solved in
	\PTIME.
\end{restatable}
\begin{proof}
	We consider the representative case of $d=3$.
	Suppose we are given a "right-resetting" "pushdown tree automaton"
	$\Aa=(Q,\S,\Gamma,q^0,\bot,\d,\W)$.

	The first step is to construct a pushdown word automaton $\Aal(G_1,G_2,G_3)$
	depending on three sets of states $G_1,G_2,G_3\incl Q$.
	The idea is that $\Aal$ simulates the run of $\Aa$ on the leftmost branch of
	a tree.
	When $\Aa$ has a transition going both to the left and to the right then $\Aal$ goes
 	to the left and checks if the state going to the right is in an appropriate $G_i$.
	This means that $\Aal$ works over the alphabet $\S^l$ that is the same as $\S$
	but all letters from $\S_2$ have rank $1$ instead of $2$.
	The states of $\Aal(G_1,G_2,G_3)$ are $Q\times\set{1,2,3}$ with the second component
	storing the maximal rank of a state seen so far on the run. 
	The transitions of $\Aal(G_1,G_2,G_3)$ are defined according to the above description.
	We make precise only the case for a transition of $\Aa$ of the form
	$\d(q,a,\g)=((q_l,\instr_l),(q_r,\instr_r))$.
	In this case, $\Aal$ has a transition
	$\d^l((q,i),a,\g)=((q_l,\max(i,\W(q_l))),\instr_l)$ if $q_r\in
	G_{\max(i,\W(q_r))}$. 	
	Observe that $\instr_r$ is necessarily $\reset$ as $\Aa$ 	is right-resetting.  

	The next step is to observe that for given sets $G_1,G_2,G_3$ we can calculate in
	\PTIME\ the set of states from which $\Aal(G_1,G_2,G_3)$ has an accepting run. 

	The last step is to compute the following fixpoint expression in the lattice
	of subsets of $Q$:
	\begin{align*}
		W =\, &\LFP X_3.\ \GFP X_2.\ \LFP X_1.\ P(X_1,X_2,X_3)\qquad\text{where}\\
		P(X_1,X_2,X_3)=\, &\set{q : \Aal(X_1,X_2,X_3)\text{ has an accepting run from $q$}}\ .
	\end{align*}
	Observe that $P:\Pp(Q)^3\to\Pp(Q)$ is a monotone function over the lattice of
	subsets of $Q$.
	Computing $W$  requires at most $|Q|^3$ computations of $P$ for different
	triples of sets of states. 

	We claim that $\Aa$ has an accepting run from a state $q$, if and only if,
	$q\in W$. The argument is presented in the appendix.
\end{proof}

\begin{proof}[Proof of Theorem~\ref{thm:main-pushdown}]
  The lower bound follows already from Theorem~\ref{thm:main-finite}.

  For the upper bound we reuse the Büchi tree automaton  $\hat{\Bb}$
  from Proposition~\ref{prop:automaton-limit}.
  This time $\hat{\Bb}$ is a "pushdown tree automaton", however it is
  "right-resetting" because processes are spawned with empty stack.
  We follow the lines of the proof of Theorem~\ref{thm:main-finite},
  building  the product of $\hat{\Bb}$ with the "regular objective"
  automaton $\Aa$, and constructing an equivalent parity, 
  "right-resetting" pushdown tree automaton $\Cc$.
  Proposition~\ref{prop:emptiness-right-reset} concludes the proof. 
\end{proof}



\section{Conclusions}

We have considered verification of parametric lock sharing systems where
processes can spawn other processes and create new locks.
Representing configurations as trees, and the notion of the limit
configuration, are instrumental in our approach.
We believe that we have made stimulating observations about this
representation. 
It is very easy to express fairness as a property of a limit configuration.
Many interesting properties, including liveness, can be formulated very naturally as properties of limit trees (cf.\
page~\pageref{page:properties}).
Moreover, there are structural conditions characterizing when a tree is a limit
configuration of a run of a given system (Lemma~\ref{lem:limitlocal}).

We expect that the parameters in Theorem~\ref{thm:main-pushdown} will be usually
quite small. 
As the dining philosophers example suggests, for many systems the maximal arity
should be quite small (cf.\ Figure~\ref{fig:philosophers}). 
Indeed, the maximal arity of the system corresponds to the tree width of the graph
where process instances are nodes and edges represent sharing a lock. 
The maximal priority will be often $3$.
In our opinion, most interesting properties would have the form ``there is a left
path such that'' or ``all left paths are such that'', and these properties need
only automata with three priorities. 
So in this case our verification algorithm is in \PTIME.


Our handling of pushdown processes is different from the literature. 
Most of our development is done for finite state processes, while the transition
to pushdown process is handled through "right-resetting" concept.
Proposition~\ref{prop:emptiness-right-reset} implies that in our context
pushdown processes are essentially as easy to handle as finite processes.

As further work it would be interesting to see if it is possible to extend our
approach to treat join operation~\cite{GawlitzaLMSW11}.
An important question is how to extend the model with some shared state and
still retain decidability for the pushdown case.


\bibliography{distr}

 \newpage
 \appendix
 \appendixtrue

 \section{Fairness}


\PropFair*

\begin{proof}
  Consider the left-to-right implication. 
  Suppose towards a contradiction that the run is not "process-fair".
  So there is a transition $\t\act{\n,a}\t'$ for some leaf $\n$. 
  Let $p$ be the process moving in $\n$, and 
  let $\t_i$ be the first configuration where $\n$ appears in $\tree$.
  We show that $p$ moves after this configuration, contradicting
  the fact that $\n$ is a leaf. 

  If $\op(a)$ is not a $\get{}$ operation then $(\n,a)$ is enabled in
  every configuration $\t_j$ for $j>i$.
  By strong fairness $p$ must move after $i$.

  A more interesting case is when $\op(a)$ is $\get{x}$ for some $x$.
  Let $\ell=L(\n)(x)$ be the lock taken by the transition. 
  As $(\n,a)$ is enabled in $\t$, we have that $\ell\not\in H(\t)$. 
  We show that this implies that process $p$ is enabled infinitely often after
  position $i$. 
  By "soundness", as $\n,a$ is enabled, $p$ cannot hold $\ell$: so $\ell\not\in H(\n)$. 
  If $\ell\in H(\t_i)$ and $\ell$ is never released afterwards then there is a
  node $\n'$ in $\t_i$ (and thus in $\t$) such that for every left descendant of $\n'$ we have
  $\ell\in H(\n')$.
  But this is impossible since we have assumed $\ell\not\in H(\t)$.
  Hence, either $\ell$ is free in $\t_i$ or it is free in some later configuration
  $\t_{i_1}$ such that $\t_{i_1}\act{\n',b} \t_{i_1+1}$ and $\op(b)=\get{y}$
  with $L(\n')(y)=\ell$. 
  So, $(\n,a)$ is enabled in $\t_{i_1}$. 
  If $\ell$ is never taken after $i_1$ then $(\n,a)$ is enabled always after
  $i_1$, and we get a move of $p$ by strong fairness as before.
  If $\ell$ is taken after $i_1$ then by the same argument as above there must
  be also a position $i_2$ when $\ell$ is released. 
  So, $(\n,a)$ is enabled in $\t_{i_2}$. 
  This argument shows that $(\n,a)$ must be enabled infinitely often after $i$,
  so by strong fairness there must be a move by $p$ after $i$.

  Consider now the right-to-left implication.
  Suppose that $\t$ is "process-fair", and the process $p$ is enabled infinitely
  often after position $i$.
  By contradiction, assume that $p$ does not move after position $i$
  and let $\n$ be the last node of $p$'s local run.
  If the action $a$ of $p$ that is enabled infinitely often is not a
  $\get{}$ then  $a$ is enabled in every $\tree_j$ with $j
  \ge i$, and $\tree \act{\n,a} \tree'$, contradicting process-fairness.
  Else, $op(a)$ is $\get{x}$ with $L(\n)(x)=\ell$. Since $a$ is
  enabled infinitely often, $\ell \notin H(\tree)$. Again we have
  $\tree \act{\n,a} \tree'$, contradicting process-fairness.
\end{proof}

 \section{Characterizing limit configurations}

\LemLimitLocal*

\begin{proof}
	For the left-to-right direction we fix a strict total order $<$ on
	$H(\tree)$ that is compatible with $\prec_H$ (for instance the
	strict order $<$ defined in the proof of Lemma~\ref{lem:limit}).
	Then we order the variables $x \in \Var{\nu}$ with
	$\vallabel(\n)(x) \in H(\tree)$ according to $<$.
	The three conditions of the lemma then follow directly.
	
	For the converse we define $\prec$ on $H(\tree)$ by
	$\ell \prec \ell'$ if for some node $\n$ with variables $x\not=
	x'$ such that $\vallabel(\n)(x)=\ell$ and
	$\vallabel(\n)(x')=\ell'$ we have $x <_\n x'$.

	We start by showing that $\prec$ is acyclic.
	Assume by contradiction that $\ell_0\prec \ell_1 \cdots \prec
	\ell_k \prec \ell_0$ is a cycle of  minimal length, so
	the locks $\ell_i \in H(\tree)$ are all distinct.
	We use indices modulo $k+1$, so $k+1\equiv 0$.
	Note that $k>1$ because of condition "F4.2".
	
	By assumption, the "scopes" of $\ell_i$ and $\ell_{i+1}$
	intersect, for every $i$.
	Since "scopes" are subtrees of $\tree$ (Remark~\ref{rem:scope})
	this means that two scopes that intersect have roots that are
	ordered by the ancestor relation in $\tree$.
	
	Assume first that $k>2$.
	Let $i$ be such that the depth of the root of the  "scope" of
	$\ell_i$ is maximal.
	So the roots of the "scopes" of $\ell_{i-1}$ and  $\ell_{i+1}$ are
	ancestors of the root $\n$ of the "scope" of $\ell_i$.
	In the "scope" of $\ell_i$ there exist nodes that belong to the
	"scope" of $\ell_{i-1}$ and of $\ell_{i+1}$, respectively.
	This means that $\n$ is in the scope of $\ell_{i-1},\ell_i$ and
	$\ell_{i+1}$.
	So the "scopes" of $\ell_{i-1}$ and $\ell_{i+1}$ intersect, and we
	have either $\ell_{i-1} <_\n \ell_{i+1}$ or $\ell_{i+1} <_\n
	\ell_{i-1}$.
	Thus we get from the definition of $\prec$ either
	$\ell_{i-1} \prec \ell_{i+1}$ or $\ell_{i+1} \prec    \ell_{i-1}$. 
	In both cases the cycle $\ell_0\prec \ell_1 \cdots \prec 
	\ell_k \prec \ell_0$ is not minimal, a contradiction.
	
	It remains to consider the case $\ell_0 \prec \ell_1 \prec \ell_2
	\prec \ell_0$.
	With a similar argument as before there exists a node $\n$ which is in
	the "scope" of all of $\ell_0,\ell_1,\ell_2$, so this node gives a
	total order on these locks and there cannot exist a cycle. 
	
	We now show that $\bef$ is acyclic as well.
	
	Like before, suppose there exists a cycle of distinct nodes $\ell_0 \bef \ell_1 \bef \cdots \bef \ell_k \bef \ell_{k+1} = \ell_0$ with $k > 0$. We consider such a cycle of minimal size. Hence every $\ell_i$ is comparable with $\ell_{i-1},  \ell_{i+1}$ and incomparable with all the other $\ell_j$ (as otherwise we would obtain a shorter cycle).
	
	Given $\ell, \ell' \in H(\t)$ such that $\ell \bef \ell'$, let
        $\nu$ be the node with an "unmatched" $\get{}$  of $\ell$. By
        condition~$F3$, this node is unique. 
	By the definition $\prec_H$ this node has some descendant
        $\nu'$ with an operation on $\ell'$. There are two
        possibilities, one is that the scopes of $\ell,\ell'$ intersect, in which case by condition~"F4.3" we have $\ell \prec \ell'$.
	The other possibility is that the two subtrees do not intersect, in which case the root of $\subtree(\ell')$ is strictly below the "unmatched" get of $\ell$. 
	
	As $\prec$ is acyclic, there exists some  $i$ such that the
        scopes of $\ell_{i-1}$ and $\ell_{i}$ are disjoint, hence all
        nodes of the scope of $\ell_{i}$ are below the "unmatched" $\get{}$
        of $\ell_{i-1}$.
        In particular the "unmatched" $\get{}$ of $\ell_{i-1}$ is an ancestor of the "unmatched" $\get{}$ of $\ell_{i}$. As a result, by the definition of $\bef$, $\ell_{i-1} \bef \ell_{i+1}$.
	
	If $k \geq 2$ then the above argument shows that the cycle was
        not minimal, 
        yielding a contradiction.
	
	If $k = 1$ then we have a contradiction as well, as 
        either the scopes intersect, so we cannot have both $\ell_0
        \prec \ell_1$  and $\ell_1 \prec \ell_0$.
        Or they do not intersect, but then there is a node $\n$ in the
        intersection with either $\ell_0 <_\n \ell_1$ or $\ell_1 <_\n
        \ell_0$, but not both.
	
	As a result, the relation $\bef$ is acyclic.
\end{proof}

 \section{Recognizing limit configurations}

\LemSemSynt*

\begin{proof}
	Suppose $\Cc$ is a correct coloring of $\t$.
	Clearly $\t$ is "process-consistent".
	Branch consistency follows from Lemma~\ref{lem:limit}
        (condition "F3").
	Indeed, all the clauses for $\keeps(x)$ and $\evkeeps(x)$ hold because of the
	third condition of this lemma.
	The clauses for $\avoids(x)$ and $\evavoids(x)$ follow directly from the
	semantics. 
	Directly from definition of the correct coloring it follows that it is also 
	"eventuality-consistent". 
	It is slightly more difficult to verify that it is "recurrence-consistent". 

	To verify recurrence consistency of $\Cc$ consider an arbitrary node $\n$ of $\t$
	and an "uncolored" variable $x\in\Var{\n}$.
	We find an infinite sequence:
	\begin{equation*}
		(\n,x)=(\n_0,x_0),(\n_1,x_1),\dots
	\end{equation*}
	such that 
	\begin{itemize}
		\item $x_i \in\Var{\n_i}$ and $x_i$ is "uncolored" in
                  $\n_i$,
                  \item $\n_0,\n_1,\dots$ is a path, 
		\item $x_i \sim x_{i+1}$.
	\end{itemize}
	Let us see why it is possible. 
	Since every leaf satisfies either $\n\sat\keeps(x)$ or $\n\sat\avoids(x)$, node
	$\n$ is not a leaf.
	If $\n0$ satisfies $\keeps(x)$ or $\evkeeps(x)$ then so does $\n$.
	If $\n0$ satisfies $\avoids(x)$ or $\evavoids(x)$ and is the unique successor
	then so does $\n$.
	Hence, if $\n0$ is the unique successor of $\n$ then $x$ cannot be colored in
	$\n0$.
	If $\n1$ exists, but there is no $y$ with $x=\s(\n1)(y)$ then the same
	verification shows that $x$ cannot be colored in $\n0$.
	Finally, if $x=\s(\n1)(y)$ and $y$ is colored in $\n1$ then $x$ must be colored
	too.
	Hence, $y$ is not colored in $\n1$ in this last case.
	This shows how to find $(\n_1,x_1)$.
	Repeating this argument we obtain the desired sequence.

	To terminate we show why the existence of the above sequence implies the recurrence condition.
	First note that $x_i \sim x_j$ 
        for all $i,j \ge 0$.
        Let $\ell=L(\n)(x)$.
	We observe that since $\n_i$ does not satisfy  $\avoids(x_i)$ then there must be
	an operation on $\ell$ below $\n_i$, and since it does not
	satisfy $\evkeeps(x)$ it must be a release. 
	So we have found an infinite path such that in the subtree of every node of
	this path there is a release operation. 
	This means that there are infinitely many get and release
	operations on $\ell$ in the tree below $\n$
\end{proof}
%

\LemSyntSem*
      
\begin{proof}
	Process consistency guarantees that locally labels follow the transition
	relations.
	Branch consistency on $\keeps(x)$ and $\avoids(x)$ labels guarantees that
	if $\n$ is labeled by one of these predicates then the predicate holds in
	$\n$.
	To get the same property for $\evkeeps(x)$ and $\evavoids(x)$ we need 
	the "eventuality-consistent" condition. 

	Finally, if $x$ is "uncolored" at $\n$ then the "recurrence-consistent"
	 condition  implies that $x$ satisfies none of the
        four predicates. 
\end{proof}

%

\LemHS*
\begin{proof}
	Suppose $\ell=L(\n)(x)$ and  $\ell\in H(\t)$. 
	Take the node $\n'$ that is closest to the root and has
        $\ell=L(\n')(x')$ for some $x'$. 
	We have $\n'\sat\evkeeps(x')$ and $\n'$ is a right child (it cannot be the
	root as $\Var{\e}=\es$).
	Hence, $x'\in H^s(\n')$. 
	By induction on the length of the path from $\n'$ to $\n$ we show that $x\in
	H^s(\n)$.

	For the other direction, if $\n\sat\evkeeps(x)$ then
        $L(\n)(x)\in H(\t)$. 
	It is also easy to see that membership in $H(\t)$ is preserved by all the rules.
\end{proof}

\LemCheckFtwo*

\begin{proof}
	By Lemma~\ref{lem:Hs}.
\end{proof}

\LemCheckFthree*
\begin{proof}
	Consider $\ell\in H(\t)$.
	By Lemma~\ref{lemma:H-from-coloring} there is a node $\n$ and $x\in\Var{\n}$
	with $\n\sat\keeps(x)$, $\ell=L(\n)(x)$.
	Let $\n'$ be the root of the "scope" of $\ell$.
	We have $\n'\sat\evkeeps(x')$ for $x' \in \Var{\n'}$ with $x'
        \sim x$.
	By consistency conditions on the coloring:
	\begin{itemize}
		\item for every node $\n''$ on the path from $\n'$ to $\n$ we have
		$\n''\sat\evkeeps(x'')$ for $x'' \sim x$, and for
                every right child $\n'''$ of $\n''$ we have
                $\n'''\sat\evavoids(x''')$ for $x'''\sim x$.
	\end{itemize}
	Observe that $\n'''\sat\evavoids(x''')$ guarantees that there are only finitely
	many operations on $\ell$ below $\n'''$, and that there is no "unmatched" get of
	$\ell$ below $\n'''$. 
	Since there are no operations on $\ell$ below $\n$, we are done.
\end{proof}

%

\LemSyntOrd*
\begin{proof}
	Let us take a family of orders $<_\n$ satisfying conditions "F4.1", "F4.2", "F4.3" of
	Lemma~\ref{lem:limitlocal}.
	We show that it is a "consistent order labeling" of $\t$.
	By Lemma~\ref{lem:Hs} the first condition is satisfied.
	The next two conditions follow from condition "F4.2".
	The fourth condition requires some verification.
	Consider $\n$ as in that condition, so with $y <_\n x$ and $\n
        \sat \keeps(x)$. 
	It follows that there is some ancestor  $\n'$ of $\n$,
        together with some $x' \sim x$, $x' \in \Var{\n'}$, such
        that  the action at $\n'$ is an "unmatched" $\get{x'}$ of the lock
	$\ell=L(\n')(x')=L(\n)(x)$.
	If there were some operation on $\ell'=L(\n)(y)$ below or at $\n$ then
	$\ell\prec_H\ell'$, implying $x<_\n y$ by "F4.3".
	Thus there is no operation on $\ell'$ below or at $\n$, meaning that
  $\n\sat\keeps(y)$ or $\n\sat\avoids(y)$.

	For the other direction, take a "consistent order labeling"
	$<_\n$.
	We show that it satisfies the conditions "F4.1", "F4.2", "F4.3" of
	Lemma~\ref{lem:limitlocal}. 
	From the first condition on $<_\n$ and
	Lemma~\ref{lem:Hs} we see that $<_\n$ orders only variables
  associated with locks from $H(\t)$; this gives us "F4.1".
	Condition "F4.2" follows directly from the second and third
  property of "consistent order labeling".

	It remains to show "F4.3".
	For this take a node $\n$ and two locks $\ell=L(\n)(x)$ and
        $\ell'= L(\n)(y)$ for some $x,y\in H^s(\n)$.
	Suppose $\ell \prec_H \ell'$.
	This means that there is an "unmatched" get of $\ell$, say in
        a node $\n'$,  and an
	operation on $\ell'$ at some node $\n''$ below $\n'$.

        We show below that we can find some node $\n_1$ in the "scope" of both
        $\ell$ and $\ell'$, and such that $\n_1 \sat \keeps(x_1)$ and
        $\n_1 \not\sat \keeps(y_1)$ and $\n_1 \not\sat \avoids(y_1)$, with $x \sim x_1$ and $y \sim y_1$.
        This will show that we cannot have $y_1 <_{\n_1} x_1$, so it
        must hold that $x_1 <_{\n_1} y_1$, thus also $x <_\n y$ by
        local consistency.

        \begin{itemize}
        \item If either $\n,\n'$ are incomparable, or $\n$ is an
          ancestor of $\n'$, or $\n=\n'$, then $\n'$ and $\n'0$ are in the "scope" of
          both $\ell$ and $\ell'$ (note that $\n'0$ is an ancestor of
          $\n''$, or they can be equal).
          We choose $\n_1=\n'0$.
          \item If $\n' \not=\n$ is an ancestor of $\n$, but $\n$ and $\n''$
            are either incomparable, or $\n$ is an ancestor of $\n''$,
            then we chose $\n_1$ as the least
            common ancestor of $\n''$ and $\n$.
            Note that $\n_1$ is below or equal to $\n0$, and belongs
            to the "scope" of both $\ell$ and $\ell'$.
            \item If $\n''$ is an ancestor of $\n$ then $\n''$ is in
              the "scope" of both $\ell$ and $\ell'$, so we chose
              $\n_1$ to be $\n''$.
        \end{itemize}

\end{proof}

\WellFounded*

\begin{proof}
	The first statement is easy: take the well-founded strict
        order on locks $<$ defined in the proof of
        Lemma~\ref{lem:limit}, and for each node $\n$ take as $<_\n$
        the order given by $<$ on $L(\n)(\Var{\n})$. The
        well-foundedness of $<$ implies that there is no "infinite
        descending chain" in the order labeling.
	
	For the second part, assume $\prec_H$ has an infinite descending chain and let $(<_\n)_{\n \in \t}$ be a "consistent order
	labeling" of $\t$. Let $\ell_0 \aft \ell_1 \aft \cdots$ be an "infinite descending chain" for $\bef$.
	
	\AP For all $i\geq 1$ let $\mu_i$ be a node with an
        "unmatched" $\get{}$ of $\ell_i$ and with a descendant with a
        $\get{}$ of $\ell_{i-1}$.
        Let $c_i$ be the root of the "scope" of $\ell_i$ in $\t$.
        As $\mu_{i+1}$ is an ancestor of a node where $\ell_{i}$ appears, it is "comparable" with $c_i$ (two nodes are ""comparable"" if one is an ancestor of the other). As $c_{i+1}$ is an ancestor of $\mu_{i+1}$, it is "comparable" with $c_i$.\medskip
	
	\noindent
	\textbf{Claim 1}: For all $a\leq b$, there exists $i \in
        \set{a,\ldots, b}$ such that $c_i$ is an ancestor of all $(c_k)_{a\leq k \leq b}$.
	\begin{proof}
		We proceed by induction on $b-a$. If $b-a=0$ this is clear.
		If $b-a>0$, by induction hypothesis there exists $i \in \set{a,\ldots, b-1}$ such that $c_i$ is an ancestor of all $(c_j)_{a\leq k \leq b-1}$.
		As $c_b$ is "comparable" with $c_{b-1}$, which is a descendant of $c_i$, $c_b$ is "comparable" with $c_i$.
		If $c_b$ is a descendant of $c_i$, then $c_i$ is an ancestor of all $(c_k)_{a\leq k \leq b}$. If $c_b$ is an ancestor of $c_i$, then $c_b$ is an ancestor of all $(c_k)_{a\leq k \leq b}$. 
	\end{proof}

	Consider the subtree of $\t$ formed by all $c_i$ and their
        ancestors. It is an infinite, but finitely-branching tree, thus it has an infinite branch by König's lemma.
	We first argue that there must be infinitely many $c_i$ on that branch. 
	Let $a \in \nats$, let $\n_a$ be the lowest ancestor of $c_a$
        on the branch. Let $\n'$ be lower on the branch than $\n_a$,
        then $\n'$ has some descendant $c_b$. Note that $\nu_a$ is the
        lowest common ancestor of $c_a$ and $c_b$. We can assume $a<b$
        (the case $b<a$ is symmetric), then by Claim 1 there is some $i \in \set{a, \ldots, b}$ such that $c_i$ is an ancestor of all $(c_k)_{a\leq k \leq b}$, and in particular of $c_a$ and $c_b$. Further, as $\nu_a$ is the lowest common ancestor of $c_a$ and $c_b$, $c_i$ is an ancestor of $\nu_a$ and is thus on the branch.
	As a result, for all $a \in \nats$ we can find $i\geq a$ such that $c_i$ is on the branch.
	
	We pick a sequence of $c_i$ as follows: we start with the
        highest $c_{i_0}$ on the branch, and then define $c_{i_{j+1}}$
        as the highest $c_i$ on the branch with $i>i_j$, for all $j$.
	By definition for all $j$ we have that no $c_i$ with $i>i_j$ is a strict ancestor of $c_{i_{j+1}}$.
	
	As a consequence of Claim 1, there exists $i \in \set{i_{j}+1, \ldots, i_{j+1}}$ such that $c_i$ is an ancestor of all $(c_k)_{i_j+1\leq k \leq i_{j+1}}$. As noted above, as $i>i_j$ we cannot have $c_i$ as a strict ancestor of $c_{i_{j+1}}$, hence $i=i_{j+1}$. As a result, $c_{i_{j+1}}$ is an ancestor of all $(c_k)_{i_j<k < i_{j+1}}$.

        For the remaining of the proof we fix a "consistent order
        labeling" $(<_\n)_\n$ for $\t$.
        \medskip

	\noindent
	\textbf{Claim 2}: For all $a < b$, if node $\n$ is in the
        "scope" of both $\ell_a$ and $\ell_b$, and if $\nu$ is
        ancestor of all $(c_k)_{a<k<b}$, then $x >_{\n} y$, with $x,y$ such that $L(\n)(x) = \ell_a$ and $L(\n)(y) = \ell_b$.
	
	\begin{proof}
		Let $x,y$ be such that $L(\n)(x) = \ell_a$ and $L(\n)(y) = \ell_b$.
		We proceed by induction on $b-a$.
		
                If $b= a +1$ then $\ell_a \aft \ell_b$.
                Since we assume that $(<_\n)_n$ is a "consistent order labeling",
                by Lemma~\ref{lem:synt-orders} and  "F4.3" of
                Lemma~\ref{lem:limitlocal} we have $x >_\n y$ as claimed.

		If $b-a\geq 2$, by Claim 1, there exists $i \in
                \set{a+1,\ldots, b-1}$ such that $c_i$ is an ancestor
                of all $(c_k)_{a<k<b}$. In particular, $c_i$ is an
                ancestor of $c_{a+1}$, itself an ancestor of
                $\mu_{a+1}$, itself an ancestor of a node $\nu'$ with
                a $\get{}$ of $\ell_a$.
                Recall that $\n$ itself is an ancestor of $c_i$, by assumption.
		As the "scope" of a lock is a subtree, $\ell_a$ appears in all nodes between $\nu$ and $\nu'$, thus in particular in $c_i$.
		
		Moreover, $\mu_b$ is an ancestor of some node 
                with a $\get{}$ of $\ell_{b-1}$, which is a descendant
                of $c_{b-1}$, thus of $c_i$, hence $\mu_b$ and $c_{i}$
                are "comparable". If $\mu_b$ is an ancestor of $c_i$,
                then as $\nu'$ is a descendant of $c_i$, $\nu'$ is
                also a
                descendant of $\mu_b$, hence $\ell_a \aft \ell_b$. As a result, $x >_\n y$
                as the "consistent order labeling"  satisfies "F4.3" (Lemma~\ref{lem:synt-orders}).
		If $\mu_b$ is a descendant of $c_i$ then as the
                "scope" of  a lock is a subtree, $\ell_b$ appears in all nodes between $\nu$ and $\mu_b$, thus in particular in $c_i$. 
		We set $x', y', z' \in \Var{c_i}$ such that
                $L(c_i)(x') = \ell_a$, $L(c_i)(y') = \ell_b$ and
                $L(c_i)(z')=\ell_i$ and by induction hypothesis we
                have $x >_{c_i} z$ and $z >_{c_i} y$ thus $x >_{c_i}
                y$ as $>_{c_i}$ is total.
                Finally, as we have a "consistent order labeling",  $x
                >_\n y$ holds as well. 
	\end{proof}
	
	Recall that $c_{i_{j+1}}$ is an ancestor of all $(c_k)_{i_j< k
          < i_{j+1}}$.
        In particular,   $c_{i_{j+1}}$ is an ancestor of $c_{i_j +1}$,
        thus of $\mu_{i_j+1}$, itself an ancestor of some node $\nu'$
        with a $\get{}$ of $\ell_{i_j}$.
        Thus $\ell_{i_j}$ appears in $c_{i_{j+1}}$, because
        $c_{i_{j+1}}$ is between $c_{i_j}$ and $\nu'$.
	
	Let $x_j, y_j$ be such that $L(c_{i_{j+1}})(x_j) = \ell_{i_j}$ and $L(c_{i_{j+1}})(y_j) = \ell_{i_{j+1}}$.
	By Claim 2, we have $x_j >_\n y_j$. The sequences $(c_{i_{j+1}})_{j>0}$, $(x_j)_{j>0}$ and $(y_j)_{j>0}$ thus form an "infinite descending chain", proving the lemma. 
\end{proof}

\PropAutLim*
\begin{proof}
	Given a tree $\t$ labeled with $p,a,s$ the automaton $\wh\Bb$ guesses a coloring $\Cc$,
	labeling $H^s$ and an ordering labeling $\Oo$. 
	It then checks if $\Cc$, $H^s$ and $\Oo$ satisfy all the consistency
	conditions. 
	This automaton is a product of the following automata:
	\begin{itemize}
		\item $\Bb_1$ recognizing "process-consistent" trees,
		\item $\Bb_\Cc$ checking if the coloring is "syntactically correct",
		\item $\Bb_H$ checking if $H^s$ is a "syntactic $H$-labeling",
		\item $\Bb_2$ checking the conditions of Lemma~\ref{lem:chekcing-F2},
		\item $\Bb_\Oo$ checking if $\Oo$ is a "consistent order labeling".
		\item $\Bb_5$ checking the absence of infinite descending chains (Lemma~\ref{lem:well-founded}).
	\end{itemize}
	Apart from $\Bb_\Cc$ and $\Bb_5$ the other automata only check relations between a node and its
	children and some additional conditions local to a node. 
	So they are automata with trivial acceptance conditions.
	Automaton $\Bb_\Cc$ needs a B\"uchi condition to check that $\Cc$ is
	"eventuality-consistent" and "recurrence-consistent".
	The number of labels is polynomial in the size of DLSS and exponential in the
	maximal arity as we have sets of predicates and orderings on variables as
	labels.
	Automaton $\Bb_5$ can be obtained by first constructing an automaton for its complement: one can easily define a non-deterministic Büchi automaton guessing a branch and following a sequence of variables along that branch witnessing an infinite decreasing sequence of locks. As it only needs to remember a pointer to one of the variables of a node, its number of states is the maximal arity of the "DLSS". Thus we can complement it to get a non-deterministic Büchi automaton checking the absence of such sequence, of size exponential in the maximal arity of the "DLSS", and polynomial in the alphabet (itself exponential in the arity and polynomial in the "DLSS").
	
	We need to check that $\t$ is a fair limit configuration if and only if it is
	accepted by $\wh\Bb$.

	If $\t$ is a limit configuration then it is process consistent, so it is
	accepted by $\Bb_1$.
	Guessing $\Cc$ to be "semantically correct" coloring ensures that $\Bb_\Cc$
	accepts $\t$ with this coloring (Lemma~\ref{lem:sem-coloring-is-synt}).
	As we have observed, given the coloring there is unique "syntactic $H$-labeling", so
	$\Bb_H$ can accept it.
	By Lemma~\ref{lem:limit}, configuration $\t$ satisfies properties "F1-5".
	So $\t$ is accepted by $\Bb_2$.
	Finally, by Lemma~\ref{lem:limitlocal}, $\t$ satisfies properties "F4.1", "F4.2",
	"F4.3", so $\t$ is accepted by $\Bb_\Oo$ thanks to
        Lemma~\ref{lem:synt-orders}.
        By Lemma~\ref{lem:well-founded}, the automaton $\Bb_5$ accepts
        $\t$ as well.

	For the other direction suppose $\t$ is accepted by $\wh\Bb$. 
	Thanks to Lemma~\ref{lem:limit} it is sufficient to check properties "F1-5".
	Property "F1" is verified by automaton $\Bb_1$.
	Thanks to $\Bb_\Cc$ we know that the guessed coloring is syntactically
	correct.
	Then $\Bb_2$ ensures that $\t$ satisfies "F2" thanks to
	Lemma~\ref{lem:chekcing-F2}.
	Lemma~\ref{lem:chekcing-F3} ensures that $\t$ satisfies "F3".
	Finally, automaton $\Bb_\Oo$ checks that the guessed orderings are a "consistent
	order labeling".
	Hence, Lemma~\ref{lem:synt-orders} guarantees that $\t$ satisfies the conditions of
	Lemma~\ref{lem:limitlocal} giving us "F4".
	Finally, by Lemma~\ref{lem:well-founded} automaton $\Bb_5$ verifies condition "F5".
      \end{proof}

 \section{Pushdown systems}

\PropRightResetting*
\begin{proof}
	We consider the representative case of $d=3$.
	Suppose we are given a "right-resetting" "pushdown tree automaton"
	$\Aa=(Q,\S,\Gamma,q^0,\bot,\d,\W)$.

	The first step is to construct a pushdown word automaton $\Aal(G_1,G_2,G_3)$
	depending on three sets of states $G_1,G_2,G_3\incl Q$.
	The idea is that $\Aal$ simulates the run of $\Aa$ on the leftmost branch of
	a tree.
	When $\Aa$ has a transition going both to the left and to the right then $\Aal$ goes
 	to the left and checks if the state going to the right is in an appropriate $G_e$.
	This means that $\Aal$ works over the alphabet $\S^l$ that is the same as $\S$
	but all letters from $\S_2$ have rank $1$ instead of $2$.
	The states of $\Aal(G_1,G_2,G_3)$ are $Q\times\set{1,2,3}$ with the second component
	storing the maximal rank of a state seen so far on the run. 
	The transitions of $\Aal(G_1,G_2,G_3)$ are defined according to the above description.
	We make precise only the case for a transition of $\Aa$ of the form
	$\d(q,a,\g)=((q_l,\instr_l),(q_r,\instr_r))$.
	In this case, $\Aal$ has a transition
	$\d^l((q,e),a,\g)=((q_l,\max(e,\W(q_l))),\instr_l)$ if $q_r\in
	G_{\max(e,\W(q_r))}$. 	
	Observe that $\instr_r$ is necessarily $\reset$ as $\Aa$ 	is right-resetting.  

	The next step is to observe that for given sets $G_1,G_2,G_3$ we can calculate in
	\PTIME\ the set of states from which $\Aal(G_1,G_2,G_3)$ has an accepting run. 

	The last step is to compute the following fixpoint expression in the lattice
	of subsets of $Q$:
	\begin{align*}
		W=\, &\LFP X_3.\ \GFP X_2.\ \LFP X_1.\ P(X_1,X_2,X_3)\qquad\text{where}\\
		P(X_1,X_2,X_3) \,=&\set{q : \Aal(X_1,X_2,X_3)\text{ has an accepting run from $q$}}\ .
	\end{align*}
	Observe that $P:\Pp(Q)^3\to\Pp(Q)$ is a monotone function over the lattice of
	subsets of $Q$.
	Computing $W$  requires at most $|Q|^3$ computations of $P$ for different
	triples of sets of states. 

	We claim that $\Aa$ has an accepting run from a state $q$, if and only if,
	$q\in W$.

	Let us look at the right-to-left direction of the claim. 
	For this we recall how the least fixpoint is calculated.
	Consider any monotone function $R(X)$ over $\Pp(Q)$, and its least fixpoint
	$R^\w=\LFP X.\ R(X)$.  
	This fixpoint can be computed by a sequence of approximations: 
	\begin{equation*}
		R^0=\es\qquad R^{i+1}=R(R^i)
	\end{equation*}
	The sequence of $R^i$ is increasing and $R^\w=R^i$ for some $i\leq |Q|$.

	Now we come back to our set $W$.
	Observe that $W=\LFP X_3. \ R$ where $R(X)=\GFP X_2.\LFP X_1.\
        P(X_1,X_2,X)$.
	As in the previous paragraph we can define 
	\begin{equation*}
		W^0=\es\qquad\text{and}\qquad W^{i+1}=\GFP X_2.\LFP
                X_1.\ P(X_1,X_2,W^{i})\ .
	\end{equation*}
	So, if $q\in W$ then $q\in W^i$ for some $i$.
	Now observe that $W^i=\LFP X_1.P(X_1,W^i,W^{i-1})$, since $W^i$ is a fixpoint
	of $\GFP X_2$.
	By similar reasoning we define 
	\begin{equation*}
		W^{i,0}=\es \qquad\text{and}\qquad W^{i,j+1}=P(W^{i,j},W^i,W^{i-1})\ .
	\end{equation*}
	Now, $q\in W^i$ implies $q\in W^{i,j}$ for some $j$.
	We write $\sig(q)$ for the lexicographically smallest $(i,j)$ such that $q\in
	W^{i,j}$.
	
	We examine what  $\sig(q)=(i,j)$ means.
	By definition $q\in P(W^{i,j-1},W^i,W^{i-1})$, so there is an accepting run of
	$\Aal(W^{i,j-1},W^i,W^{i,j-1})$ from $q$.
	Looking at the run of $\Aa$ that $\Aal$ simulates we can see that whenever
	this run branches to the right with some $q'$ and $e$ is the maximal rank on
	the run till this branching then 
	\begin{itemize}
		\item if $e=1$ then $q'\in W^{i,j-1}$,
		\item if $e=2$ then $q'\in W^{i,k}$ for some $k$,
		\item if $e=3$ then $q'\in W^{i-1,k}$ for some $k$.
	\end{itemize}

	With this observation we can construct an accepting run of $\Aa$ from every
	state in $W$.
	If $\sig(q)=(i,j)$ then consider an accepting run of $\Aal$ on the left path given by
	$P(W^{i,j-1},W^i,W^{i-1})$.
	For every state branching to the right from this left path we recursively apply
	the same procedure.
	By construction, every path that is eventually a left path is accepting.
	A path branching right infinitely often is also accepting by the previous
	paragraph since signatures cannot go below $0$.
        More precisely, the path cannot see $3$ infinitely often
        because the first component of the signature decreases.
        If it sees $1$ infinitely often then it needs to see also $2$
        infinitely often, because of the second component that decreases.

	Let us now look at the left-to-right direction.
	Take an accepting run of $\Aa$ from $q^0$.
	We construct something that we call a skeleton tree of this run. 
	As the nodes of the skeleton tree we take the root and all the nodes that are
	a right child; so these are the nodes of the tree of the form $(0^*1)^*$.
	The skeleton has an edge $\n\act{e}\n0^k1$ if $e$ is the maximal rank of a state of $\Aa$ on the
	path from $\n$ to $\n0^k1$.
	Observe that a node can have infinitely many children. 
	As we have started with an accepting run, every path in this skeleton tree satisfies
	the parity condition. 
	In particular, for every node, on every path from this node there is a finite
	number of $3$ edges. 
	Thus, to every node $\n$ we can assign an ordinal $\th^3(\n)$ such that if $\n\act{e}\n'$
	then $\th^3(\n')\leq\th^3(\n)$ and the inequality is strict if $e=3$.
	It is also the case that on every path from $\n$ there is a finite number of
	$1$ edges before some $2$ or $3$ edge. 
	This allows to define $\th^1(\n)$ with the property that if $\n\act{1}\n'$ then $\th^1(\n')<\th^1(\n)$.
	Now we can show that for every node $\n$ of the skeleton tree, if $q$ is the
	state assigned to $\n$ then $q\in W^{\th^3(\n),\th^1(\n)}$ for $W^{i,j}$ as defined
	in the computation of $W$ (putting $W^{\th^3}=W$ for every $\th^3>|Q|$, and
	$W^{\th^3,\th^1}=W^{\th^3}$ for every  $\th^1 >|Q|$).
	The proof is by induction on the lexicographic order on
	$(\th^3(\n),\th^1(\n))$.
\end{proof}

\section{Lower bounds}

In this section we show the two remaining lower bounds, namely
\EXPTIME-hardness for "nested" "DLSS" and undecidability for arbitrary
ones. 

\PropExpHard*

\begin{proof}
	
	We show that the difficulty of the problem stems from the
        systems and not the specification, by proving that checking if
        some copy of a process has an infinite run is already
        \EXPTIME-hard. 
	
	We provide a reduction from the emptiness problem for the
        intersection of top-down tree automata (over finite trees).
	Let $\aut_1, \ldots, \aut_k$ be finite tree automata over a ranked alphabet $A$, with $\aut_i = (S_i, \delta_i, s_{0,i}, F_i)$.
	We can assume that $A = \set{a,b,c}$ with $a,b$ of arity 2 and
        $c$ of arity 0, and that all automata only recognize trees
        with root labelled by $a$.
	We are going to construct a "DLSS" that simulates their
        computations simultaneously on the same tree $T$, by using locks to memorize their states.
	
	The  idea is to have a new process copy for each node of $T$.
        Each such copy uses the variables $x_i^s$ and
        $y_i^s$ for all $1 \leq i \leq k$ and $s \in S_i$, as well as
        the variables $z_1, z_2$ and $t$.
        Variable $x_i^s$ is supposed to encode the information about
        the state of the parent node in the run of  $\aut_i$, while
        $y_i^s$ will encode the state of $\aut_i$ at the current node.
	
	We use processes $p_0, q, \ch, \tk$, plus processes $p_a^0,
        p_a^1, p_b^0, p_b^1$ and $p_c$. 
	The processes are sketched in Figure~\ref{fig:EXPhard}.
        Process $p_0$ is  the initial one.
        Process $q$ is the root of $T$, and after spawning
        its children it takes and releases the lock $\ell_t$ associated
        with variable $t$ indefinitely.
        Lock $\ell_t$ will be shared with all processes, except for
        $\tk$. 
	The specification is that $q$ should keep running forever.
        Equivalently, lock $\ell_t$ should be free infinitely often. 
	
	Whenever process $\ch$ is spawned, the purpose is to check if
        some lock is "ultimately held".
        The first $\get{}$ in process $\ch$ corresponds to this test.
        If the lock associated with variable $s$ is not taken, by
        "process-fairness" $\ch$ 
        eventually takes it, and then it also takes  $\ell_t$
        (forever), preventing $q$ to have an infinite run.
	As for $\tk$, it simply takes the lock that it is given as argument.

	The processes representing nodes of $T$ are described next.
        Each node of $T$ is represented by a copy of process $p_a^i$
        or  $p_b^i$, depending on the letter $a$ or $b$ its parent is
        labelled with, and on whether it is a left or right child of
        its parent ($i=0$ and $i=1$, respectively).
        The root is represented by process $q$.

        Process $p_a^1$ proceeds as follows: it chooses for each of
        its children a letter and spawns the associated processes,
        with each variable $x_i^s$ of the child mapped to its own  variable
        $y_i^s$, and all $y_i^s$ of the child mapped to $\new$ (see
        actions $\spw$ in Figure~\ref{fig:EXPhard}).
        If the node represented by $p_a^1$ is a leaf, then $p_a^1$
        spans a unique child $p_c$.
        For each $1 \leq i \leq k$, process $p_a^1$ then guesses the
        state $s$ of its parent in $\aut_i$ and spawns a process $\ch$
        in charge of checking the guess, so whether the lock
        associated with $x_i^s$ is
        taken  (see actions $\b$ in Figure~\ref{fig:EXPhard}).
        Finally, $p_a^1$ spawns a copy of process $\tk$ in charge of taking
        the lock associated with $y_i^{s'}$, where $s' = \d_i(s, a,
        1)$ is the state of the current node
        (see actions $\g$ in Figure~\ref{fig:EXPhard}). 

	Processes $p_a^0$, $p_b^0, p_b^1$ are defined similarly.
        Process $p_c$ simply guesses a final state $s$ of its parent in
        each $\aut_i$  and spawns a copy of $\ch$  to verify the guesses.

        Process $q$ is also in charge of representing the root, which
        we assumed to be labelled by $a$, hence all it has to do is
        spawn two children $p_a^0$ and $p_a^1$ with the variable
        assignments described below, and spawn copies of  $\tk$
        to take the locks associated with the variables $x_i^{s_{0,i}}$ 
        (actions $\g^i$). 
	
	The problem is that  we might end up producing an infinity of
        processes, representing a computation of the automata $\aut_i$ over an
        infinite tree.  
	To avoid that, we use variables $z_0, z_1$ and $z$.
        Each copy of $p_a^i$ and $p_b^i$, after doing all its other
        operations, takes $z_0$ and $z_1$ forever, and then takes and
        releases $z$.
        When spawning other processes representing nodes, each of $p_a^i$ and
        $p_b^i$  maps the
        $z$ of the spawned process to its $z_i$ (depending on whether
        it is its first or second spawned process), and the $z_0, z_1$
        to $\new$. 
	We also add edges taking $t$ forever that will eventually be
        executed, due to "process-fairness", if $z$ is taken  forever
        before this process can take and release it. 
	Hence all such processes must first acquire the locks of $z_0,
        z_1$ and then use the lock of $z$.
        This imposes that this part of the run is executed in the
        current process after it is executed in both  children. 
		This is only possible if the  tree $T$ is finite.	
	\begin{figure}[hbt]
		\begin{tikzpicture}[node distance=1.8cm,auto,>= triangle
	45,scale=.6]
	\tikzstyle{initial}= [initial by arrow,initial text=,initial
	distance=.7cm, initial where= left]
	\tikzstyle{accepting}= [accepting by arrow,accepting text=,accepting
	distance=.7cm,accepting where =right]
	
	\node[state, minimum size=15pt,initial] (p0s0) at (0,0) {};
	\node[state, minimum size=15pt] (p0s1) at (3,0) {};
	\node (p0) at (-1,-1) {\large $p_0$};
	
	\node[state, minimum size=15pt,initial] (qs0) at (7,0) {};
	\node[state, minimum size=15pt] (qs1) at (9,0) {};
	\node[state, minimum size=15pt] (qs2) at (11,0) {};
	\node (dots2) at (12,0) {\large $\cdots$};
	\node[state, minimum size=15pt] (qs3) at (13,0) {};
	\node[state, minimum size=15pt] (qs4) at (15,0) {};
	\node[state, minimum size=15pt] (qs5) at (17,0) {};
	\node (q) at (6,-1) {\large $q$};
	
	\node[state, minimum size=15pt,initial] (takes0) at (0,-3) {};
	\node[state, minimum size=15pt] (takes1) at (3,-3) {};
	\node (tk) at (-1,-4) {\large $\tk$};
	
	\node[state, minimum size=15pt,initial] (checks0) at (8,-3) {};
	\node[state, minimum size=15pt] (checks1) at (11,-3) {};
	\node[state, minimum size=15pt] (checks2) at (15,-3) {};
	\node (ch) at (7,-4) {\large $\ch$};
	
	\node[initial, state, minimum size=15pt] (ps1) at (-0.2,-6) {};
	\node[state, minimum size=15pt] (ps1bis) at (3,-6) {};
	\node[state, minimum size=15pt] (ps1ter) at (3,-5) {};
	\node[state, minimum size=15pt] (ps2) at (6,-6) {0};
	\node[state, minimum size=15pt] (ps3) at (8,-6) {0,s};
	\node[state, minimum size=15pt] (ps4) at (10,-6) {1};
	\node (dots2) at (11.2,-6) {\large $\cdots$};
	\node[state, minimum size=15pt] (ps5) at (12.5,-6) {k-1};
	\node[state, minimum size=15pt] (ps5bis) at (15.5,-6) {k-1,s};
	\node[state, minimum size=15pt] (ps5ter) at (18.5,-6) {k,s};
	\node[state, minimum size=15pt] (ps6ant) at (11.5,-10) {};
	\node[state, minimum size=15pt] (ps6) at (14,-10) {};
	\node[state, minimum size=15pt] (ps7) at (16,-10) {};
	\node[state, minimum size=15pt] (ps8) at (18,-10) {};
	\node[state, minimum size=15pt] (ps9) at (20.5,-10) {};
	\node (pia) at (-1,-7) {\large $p_a^1$};
	
	\node[initial, state, minimum size=15pt] (pc1) at (0,-10) {0};
	\node[state, minimum size=15pt] (pc3) at (2,-10) {0,s};
	\node[state, minimum size=15pt] (pc4) at (4,-10) {1};
	\node[state, minimum size=15pt] (pc5) at (6,-10) {1,s};
	\node (dots2) at (7.3,-10) {\large $\cdots$};
	\node[state, minimum size=15pt] (pc6) at (8.7,-10) {k-1,s};
	\node (pia) at (-1,-11) {\large $p_c$};
	
	\path[->] 	
	(p0s0) edge node[above= 12pt] {$\spawn{q, \s}$} (p0s1)
	(qs0) edge node[above= 12pt] {$sp_0^a, sp_1^a$} (qs1)
	(qs1) edge node[above= 12pt] {$\g^{1}$} (qs2)
	(qs3) edge node[above= 12pt] {$\g^{k}$} (qs4)
	(takes0) edge node {$\get{s}$} (takes1)
	(checks0) edge node {$\get{s}$} (checks1)
	(checks1) edge node {$\get{t}$} (checks2)
	
	(pc1) edge node[above] {$\b^s_c$} (pc3)
	(pc3) edge node[above] {$\nop$} (pc4)
	(pc4) edge node {$\beta_c^s$} (pc5)

	(ps2) edge node[above] {$\b_{a,1}^{1,s}$} (ps3)
	(ps3) edge node[above] {$\g_{a,1}^{1,s}$} (ps4)
	(ps5) edge node[above] {$\b_{a,1}^{k,s}$} (ps5bis)
	(ps5bis) edge node[above] {$\g_{a,1}^{k,s}$} (ps5ter)
	(ps6ant) edge node[below] {$\get{z_0}$} (ps6)
	(ps6) edge node[below] {$\get{z_1}$} (ps7)
	(ps7) edge node[below] {$\get{z}$} (ps8)
	(ps8) edge node[below] {$\rel{z}$} (ps9)
	(ps1) edge node[below] {$sp_0^{b}$} (ps1bis)
	(ps1bis) edge node[below] {$sp_{1}^{b}$} (ps2)
        (ps5ter) edge node[above] {$\nop$} (ps6ant);
	\path[->, bend right=20] 
	
	(qs5) edge node[above] {$\rel{t}$} (qs4)
	(qs4) edge node[below] {$\get{t}$} (qs5)
	;
	\path[->, bend right=40] 	
	(ps1) edge node[below] {$sp^{c}$} (ps2)
	(ps6ant) edge node[below] {$\get{t}$} (ps9)
	;
	\path[->, bend left=40] 	
	(ps6) edge node[above] {$\get{t}$} (ps9)
	;
	\path[->, bend left=20] 
	(ps1) edge node {$sp_0^{a}$} (ps1ter)
	(ps1ter) edge node {$sp_{1}^{a}$} (ps2)
	;
\end{tikzpicture}
		\caption{"DLSS" in Proposition~\ref{prop:EXPTIME-hard}}
		\label{fig:EXPhard}
	\end{figure}
	
	To sum up, if the $\aut_i$ all accept a finite tree $T$ then
        we can construct a "process-fair" run of this "DLSS" by
        spawning all the processes representing  nodes of $T$,
        then having all processes $\tk$ execute their $\get{}$.
        There is no conflict as no two copies of $\tk$ take the same lock.
	All copies of process $\ch$  are then stuck in their first
        state.
        We then execute the actions on $z, z_0,z_1$ of each $p_a^i,
        p_b^i$ in a bottom-up fashion, so that we can execute them
        all. 
	Finally, we run $q$ forever by having it take and release the
        lock $\ell_t$ indefinitely.
	
	Conversely, if this "DLSS" has a "process-fair" run where $q$ runs
        forever, then  we can construct a tree $T$ over $a,b,c$ by taking
        $q$ as root and defining the children of a node as the
        processes $p_a^i, p_b^i$ spawned by the corresponding
        process.
        A node whose children are $p_a^i$ is labelled $a$, one whose
        children are $p_b^i$ is labelled $b$, and the leaves are
        labelled $c$. 
	
	We know that $T$ is finite thanks to the previous
        argument involving the $z, z_0, z_1$.
        We can
        associate with each node of $T$ a state of each $\aut_i$, inferred
        from the set of locks held so that we obtain runs of the
        automata $\aut_i$ over $T$ (transitions are respected as
        otherwise some $\ch$ process would get a lock that is not
        taken by any $\tk$ and thus would eventually take $\ell_t$
        forever). 
	This run is furthermore accepting by definition of processes $p_c$.

	The "DLSS" we constructed is clearly "nested", which shows the claim.
\end{proof}



\medskip


\ThmUndec*
\begin{proof}
  
The proof idea is to simulate an accepting run of TM $M$ using $n$ cells by
spawning a chain of $n$ processes, $P_0,\dots,P_{n-1}$.
We assume that $M$ accepts when the head is leftmost.

The initial process $P_0$ uses three locks, called $a,b,c_1$, and acquires
$a,b$ before spawning $P_1$.
Process $P_1$ uses locks $a,b,c_1$, plus a fresh lock $c_2$.
It acquires $c_1$ before spawning $P_2$.
More generally, process $P_k$ ($1\le k <n-1$) uses locks
$a,b,c_k,c_{k+1}$, and it acquires $c_k$ before spawning the next
process $P_{k+1}$.
The last process $P_{n-1}$ uses only three locks, $a,b,c_{n-1}$.

A configuration $(p,k,A_0 \dots A_{n-1})$ of the TM corresponds to
each $P_j$ storing $A_j$, with process $P_k$ storing in addition state
$p$.
A TM step to the right, from cell $k$ to $k+1$, needs to communicate
the next state $q$.

In the following we denote the process that  currently owns locks $a$
and $b$, as ``sender''.
The notation $S^+,S^-$ used below indicates that the sender tries to
send the state to the right or left neighbour, respectively.
Similarly, $R^+,R^-$ indicates that a ``receiver'' is ready to receive
from the right or the left neighbour, respectively.

Sending $q$ from $P_k$ to $P_{k+1}$ is implemented by $P_k$ using the
following sequences of actions:
\begin{eqnarray*}
  S^+_a &=& \rel{a} \get{c_{k+1}} \rel{b} \; \get{a} \rel{c_{k+1}}
            \get{b}\\
  S^+_b &=& \rel{b} \get{c_{k+1}} \rel{a} \; \get{b} \rel{c_{k+1}}
            \get{a}
\end{eqnarray*}
Process $P_{k+1}$ (``receiver'') uses matching sequences:
\begin{eqnarray*}
   R^-_a &=& \get{a} \rel{c_{k+1}} \get{b} \; \rel{a} \get{c_{k+1}}
            \rel{b}\\
  R^-_b &=& \get{b} \rel{c_{k+1}} \get{a} \; \rel{b} \get{c_{k+1}}
            \rel{a}
\end{eqnarray*}

Suppose now that the sender $P_k$ wants to send state $q$ to receiver
$P_{k+1}$.
This will be done by $P_k$ by trying to execute the sequence
$(S^+_a)^q\, S^+_b$.
Every process $P_j$ with $j>k$ is ready to execute
either $R^-_{a}$ or $R^{-}_b$.
Symmetrically, every process $P_j$ with $j<k$ is ready to execute
either $R^+_a$ or $R^+_b$.

We show next that the "DLSS"  deadlocks if $P_k, P_{k+1}$ do
not execute $(S^+_a)^q\, S^+_b$ and $(R^-_a)^q\, R^-_b$, resp., in
lockstep manner:


\medskip

\noindent
 \emph{Claim.} Assume that $P_k$ owns $\{a,b\}$, every $P_j$, $j<k$, owns
  $c_{j+1}$, and every $P_j$, $j>k$, owns $c_j$.
 Moreover, $P_k$ wants to send $a$ to $P_{j+1}$.
  Then either $P_k,P_{k+1}$ execute $S^+_a$ and $R^-_a$, resp., in
  lockstep manner, or all processes deadlock.

  \medskip

  \begin{proof}[Proof of claim]
    Process $P_k$ is the only process who can start, since
    all other processes wait for acquiring either $a$ or $b$.
    
    After releasing $a$, process $P_k$ needs $c_{k+1}$.
    It can only proceed and take $c_{k+1}$ if $P_{k+1}$ starts
    executing $R^-_a$, taking $a$ and releasing $c_{k+1}$.
    Then $P_k$ releases $b$, and waits to get back $a$.
    If $b$ is taken by another process than the receiver, say $P_j$,
    $j \not= k+1$, then
    $P_j$ will release its lock $c \not=c_{j+1}$, and $c$ is now the
    only available lock.
    Lock $a$ will never become available because $P_{j+1}$
    will not release it, so all processes deadlock.

    Assume that $P_{j+1}$ takes $b$, and releases $a$.
   If $a$ is taken by another process than the sender, say $P_j$, $j
   \not= k$, then   $P_j$ will release its lock $c \not=c_{j+1}$, and
   $c$ is now the only available lock.
   Lock $a$ will never become available because $P_{j+1}$ does not
   release $b$, so all processes deadlock.

   Assume that $P_j$ takes $a$ back.
   Then it releases $c_{j+1}$, which can be taken only by $P_{j+1}$,
   who releases also $b$.
   If $b$ is taken by another process than the sender, say $P_j$, $j
   \not= k$, then   $P_j$ will release its lock $c \not=c_{j+1}$, and
   $c$ is now the only available lock.
   Lock $b$ will never become available because $P_{j}$ does not
   release $a$ anymore.
   Once again, all processes deadlock.
  \end{proof}
We conclude the proof by noting that $P_0$ reaches a final state of
$M$ if and only if $M$ accepts.
\end{proof}

\end{document}